\documentclass[11pt]{amsart}

\usepackage{amsmath, amsthm, latexsym, amssymb, amsfonts, epsf, xcolor}
\usepackage[hidelinks]{hyperref} 
\usepackage[biblabel]{cite}
\usepackage{caption, subcaption, enumerate, color}
\usepackage{multirow,multicol}
\usepackage{mathtools}
\usepackage{bm}
\usepackage{blkarray}

\usepackage{booktabs}
\usepackage{framed}
\usepackage{setspace}
\usepackage{tikz}
\usepackage{graphicx}
\usepackage{wrapfig}
\usepackage{float}
\usepackage[dvips]{epsfig}
\usepackage{bbm}
\usepackage{url}
\usepackage{booktabs}
\usepackage{bibentry}
\usepackage[normalem]{ulem}
\usepackage{comment}
\usepackage{lingmacros}
\usepackage{cleveref}
\usepackage{extarrows}
\usepackage{multirow}
\usepackage{mathtools}
\usepackage{enumerate}
\usepackage{array}
\usepackage{bigstrut}
\usepackage{exscale,relsize}
\usepackage{graphicx}
\usepackage[utf8]{inputenc}
\usepackage[linesnumbered, ruled]{algorithm2e}

\definecolor{fporange}{RGB}{255,175,60}
\definecolor{fpgreen}{RGB}{100,200,100}
\definecolor{fpblue}{RGB}{100,170,255}
\definecolor{cellempty}{RGB}{230,240,250}

\usepackage{standalone}
\usepackage{tikz}
\usepackage{tikz-3dplot}

\DeclarePairedDelimiter\abs{\lvert}{\rvert}%

\newcommand{\gridcell}[4]{%
  \fill[#3] (#1,#2) rectangle ++(1,1);%
  \draw[gray!70] (#1,#2) rectangle ++(1,1);%
  \node at ({#1+0.5},{#2+0.5}) {#4};%
}

\newcommand{\cal}[1]{\mathcal{#1}}

\newcommand{\cM}{\cal M}

\newcommand{\cX}{\cal X}

\newcommand{\F}{{\mathbb F}}
\newcommand{\Fq}{{\mathbb F}_q}
\newcommand{\fq}{{\mathbb F}_q}
\newcommand{\N}{{\mathbb N}}

\newcommand{\xx}{{\bf x}}

\newcommand{\set}[1]{\left\{#1\right\}}

\DeclareMathOperator{\ini}{in}
\DeclareMathOperator{\simplex}{simplex}

\numberwithin{equation}{section}
\newtheorem{theorem}{Theorem}[section]
\newtheorem{lemma}[theorem]{Lemma}
\newtheorem{proposition}[theorem]{Proposition}
\newtheorem{corollary}[theorem]{Corollary}

\theoremstyle{definition}
\newtheorem{definition}[theorem]{Definition} 
\newtheorem{remark}[theorem]{Remark}
\newtheorem{example}[theorem]{Example}

\newcommand{\rmv}[1]{}

\DeclareMathOperator{\wt}{wt}

\DeclareMathOperator{\supp}{supp}
\DeclareMathOperator{\RS}{RS}
\DeclareMathOperator{\RM}{RM}

\DeclareMathOperator{\Per}{Per}
\DeclareMathOperator{\CAP}{CAP}
\DeclareMathOperator{\Car}{Car}
\DeclareMathOperator{\ev}{ev}
\DeclareMathOperator{\wH}{w_H}

\begin{document}


\title[Reed-Muller type codes over a combinatorial simplex:\\ an algebraic description]{Reed-Muller type codes over a combinatorial simplex:\\ an algebraic description}

\author[H. L\'opez]{Hiram H. L\'opez}
\author[Rodrigo]{Rodrigo San-José} 
\author[N. Shalqini]{Nart Shalqini}
\address[Hiram H. L\'opez]{Department of Mathematics\\ Virginia Tech\\ Blacksburg, VA USA}
\email{hhlopez@vt.edu}
\address[Rodrigo San-José]{Department of Mathematics\\ Virginia Tech\\ Blacksburg, VA USA}
\email{rsanjose@vt.edu}
\address[Nart Shalqini]{Department of Mathematics\\ Virginia Tech\\ Blacksburg, VA USA}
\email{nart@vt.edu}

\thanks{The authors were partially supported by the NSF grant DMS-2401558 and the Commonwealth Cyber Initiative. Hiram H. L\'opez was also partially supported by the NSF grant DMS-2502705. Rodrigo San-José was also partially supported by Grant PID2022-138906NB-C21 funded by MICIU/AEI/10.13039/501100011033 and by ERDF/EU}

\begin{abstract}
Given an ordered set $B$ of a finite field, a combinatorial simplex over $B$ is defined as the set of vectors such that the positions of the entries, with respect to $B$, sum up to a fixed integer. CAP codes are Reed-Muller type codes defined over a combinatorial simplex. They were recently introduced by Kopparty et al. as a high-rate alternative to classical Reed-Muller codes, capable of achieving arbitrarily high rates close to one for any fixed minimum distance. In this paper, we use tools from commutative algebra to analyze a combinatorial simplex and its associated CAP code. We give a universal Gr\"obner basis for the vanishing ideal of a combinatorial simplex. We describe the generalized Hamming weights of a CAP code in terms of the footprint of the vanishing ideal. For the minimum distance case, we proved a closed formula. We give a set of polynomials whose evaluations on the combinatorial simplex generate the dual of the CAP code. We describe the affine permutations that leave invariant a combinatorial simplex and use this information to prove that, in some cases, the permutation group of a CAP code is a symmetric group.
\end{abstract}

\maketitle

\section{Introduction} \label{S:intro}
Let $S \subseteq \F_q^m$ be a set of points over a finite field $\fq$, and let $V\subset \fq[x_1,\dots,x_m]$ be a set of polynomials. Evaluation codes are obtained by evaluating the set of polynomials in $V$ at the points of $S$. In particular, given $\nu\geq 0$, Reed-Muller codes are obtained when $S=\fq^m$ and $V$ is given by all the polynomials of degree at most $\nu$ \cite{kasamiRM}. 
There are many generalizations of Reed-Muller codes in the literature, e.g., evaluating in Cartesian sets \cite{LOPEZ201913}, or evaluating homogeneous polynomials in the projective space \cite{sorensen}. When $V$ is fixed to be polynomials of at most degree $\nu$ (resp., homogeneous polynomials of degree exactly $\nu$ when working with the projective space), the corresponding codes are called Reed-Muller type (resp., projective Reed-Muller type) codes, and have been studied extensively in the literature \cite{Lopez2014-bq,nestedcartesian,geramita,villarreal_gorenstein_prm_type}. 

In \cite{kopparty_high_rate}, Kopparty, Kumar, and Sha  introduce a new class of evaluation codes where the evaluation domain is restricted to an $m$-dimensional simplex. These codes, called combinatorial array of polynomial (CAP) codes, achieve a high rate while maintaining a constant minimum distance. In this paper, we further investigate the theoretical properties of these codes using methods and techniques from commutative algebra and Gr\"{o}bner basis theory. 

The generalized Hamming weights (GHWs) of a linear code are a set of parameters introduced by Wei \cite{wei1991ghw} to characterize the security of the wiretap channel of type II when using coset encoding. Over time, GHWs and their generalizations have found many applications, e.g., to list decoding \cite{guruswammiGHWlistdecoding,guruswammiGHWlistdecodingTensorInterleaved}, linear ramp secret-sharing schemes \cite{luoPropertiesRGHWs,matsumotoRGHW}, or quantum codes \cite{hamadaSteaneEnlargement,kkks}. Although there exist algorithms to compute them \cite{sanjoseGHWsPackage}, their computation is, in general, NP-hard, since the minimum distance corresponds to the first GHW and its computation is NP-hard \cite{vardyIntractability}. Nevertheless, we know the GHWs for some families of evaluation codes, such as Reed-Muller codes \cite{pellikaanGHWRM}, affine Cartesian codes \cite{beelenGHWcartesian}, Hermitian codes and norm-trace codes \cite{munuera1999hermitian,sanjoseGHWNT}, hyperbolic codes \cite{eduardoGHWHyperbolic}, some matrix-product codes \cite{sanjoseGHWMPC}, or square-free codes \cite{GHWsToricSquarefree,jaramillo,sanjoseGHWsquarefree}. 

Two other important aspects related to our work are the dual codes and the permutation group of a linear code. The dual and parity-check matrices of any family of codes are ubiquitous in many applications of coding theory, e.g., they are essential for syndrome-based decoding algorithms \cite{fengraoMajority}, and are required for applications such as secret sharing schemes \cite{chen2007secure} or quantum codes \cite{kkks}. The permutation group of a code has been extensively studied because of its different applications, such as fault-tolerant quantum computation \cite{Grassl_Roetteler}; polar coding \cite{bioglio,  1Johannsen_polar, polar_19}, proving that some codes achieve capacity \cite{RM_CA1, RM_CA2}; and code-based cryptography \cite{LESS_is_even_more, LESS}.

The structure of the paper is as follows. In Section \ref{s:vanishing_ideal}, we study the vanishing ideal of the set of evaluation points of CAP codes, which is crucial for the following sections. We derive the GHWs of CAP codes in Section \ref{s:ghws}, and recover the minimum distance as a special case.  Our Gröbner-basis approach yields the dimension and minimum distance in arbitrary characteristic in a uniform way; the bound on the minimum distance also follows from \cite{kopparty_high_rate} via a Schwartz–Zippel-type argument. In Section \ref{s:dual}, we describe the dual code of a CAP codes as the puncturing of a Cartesian code, but also as an evaluation code. Finally, we describe the permutation group of CAP codes in some cases by studying the group of affine permutations.

\section{Preliminaries}
In this section, we introduce notation and preliminary results used in the following sections. Let $\F_q$ be a finite field with $q$ elements. An $[n,k,d]$ {\bf code} over $\Fq$ is an $\F_q$-vector space $C \subseteq \F_q^n$ of dimension $k$ and minimum distance
\(d:=\min \left\{\wt(c): c \in C, c \neq 0 \right\};\)
here, $\wt(c)$ denotes the {\bf Hamming weight}, defined as the number of nonzero entries of $c$. An extension of this concept is the generalized Hamming weights (GHWs) of a linear code \cite{wei1991ghw}, which require the notion of support. We define the {\bf support} of a set $D\subseteq \fq^n$ as
$$
\supp(D):=\left\{ 1\leq i \leq n \ : \ \exists \: c \in D \textnormal{ with } c_i \neq 0 \right\}.
$$
Then, for any $1\leq r \leq k$, the $r$-th {\bf generalized Hamming weight} of $C$ is defined by 
$$
d_r(C):=\min \left\{ \abs{ \supp(D) } \ : \ D \textnormal{ is a subcode of } C \textnormal{ of dimension } r \right\}. 
$$
The {\bf dual} of $C$ is given by
\[C^{\perp} : = \{ w \in \F_q^n \ : \ w \cdot c = 0 \text{ for all } c \in C \}, \]
where $w \cdot c$ denotes the standard Euclidean inner product. 

The \textbf{shortening} of $C$ in $\{i\}$, denoted by $C_{\{i\}}$, is the code
$$
C_{\{i\}} := \{ (c_1,\dots,c_{i-1},c_{i+1},\dots,c_n) \ : \ (c_1,\dots,c_{i-1},0,c_{i+1},\dots,c_n)\in C\}. 
$$
The \textbf{puncturing} of $C$ in $\{i\}$, denoted by $C^{\{i\}}$, is the code
$$
C^{\{i\}} := \{ (c_1,\dots,c_{i-1},c_{i+1},\dots,c_n) \ : \ (c_1,\dots,c_{i-1},c_i,c_{i+1},\dots,c_n) \in C, \text{for some } c_i \in \F_q \}. 
$$
For $S\subset[n]$, we write $C_S$ (resp., $C^S$) for the successive shortening (resp., puncturing) of $C$ in the coordinates indexed by the elements in $S$. 

Let $B := \left\{b_0, b_1, \dots, b_{\ell-1} \right\} \subseteq \F_q$
be an ordered set of $\ell$ different elements of $\F_q$. Denote by $\F_q[x]_{< k}$ the set of single-variable polynomials of degree less than $k$ with coefficients in $\F_q$. A {\bf Reed-Solomon} (RS) code is denoted and defined by
\[\RS(B,k) := \left\{ \left(f(b_0),\ldots,f(b_{\ell-1})\right) \ : \ f(x) \in \F_q[x]_{< k} \right\}.\]

\begin{definition}\label{26.03.12}
For positive integers $m$ and $s \leq \ell$, we define the {\bf combinatorial simplex}, or just {\bf simplex}, as
\[
B(m,s) := \left\{(b_{i_1}, b_{i_2}, \dots, b_{i_m}) \ : \ i_1+\cdots+i_m <s, \ i_j \in \N \right \}.
\]
The set $B(m,s)$ is also called an {\bf $m$-dimensional simplex of side length $s$}. When $s=\ell$, we simplify the notation to $B(m)$.
\end{definition}
\begin{example}
We have $B(1,3) = \{b_0, b_1, b_2 \}$; see Figure~\ref{26.03.10}.
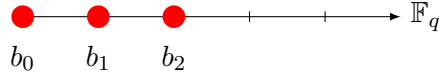
\begin{figure}[h]
\begin{center}
\begin{tikzpicture}[scale=1.]

\draw [-latex] (8,0) -- (13,0)node[right] {$\F_q$};
\foreach \x in {8,...,12} \draw (\x cm,-2pt) -- (\x cm,2pt);

\foreach \x in {1,2,3} \fill [color=red](7+\x,0)node[below]{} {circle(.15cm)};
\foreach \x in {1} \fill [color=red](7+\x,-0.25)node[below]{\color{black} $b_0$} {circle(.0cm)};
\foreach \x in {2} \fill [color=red](7+\x,-0.25)node[below]{\color{black} $b_1$} {circle(.0cm)};
\foreach \x in {3} \fill [color=red](7+\x,-0.25)node[below]{\color{black} $b_2$} {circle(.0cm)};
\end{tikzpicture}
\end{center}
    \caption{Set $B(1,3)$, a $1$-dimensional simplex of side length $3$.}
    \label{26.03.10}
\end{figure}

We also have that $B(2,3) = \left\{(b_0,b_0), (b_0,b_1), (b_0,b_2), (b_1,b_0), (b_1,b_1), (b_2,b_0) \right\}$. If we define the order $b_0 < b_1 < b_2$, then $B(2,3)$ can be seen as a triangle in $\F_q^2$; see Figure~\ref{26.03.11}.
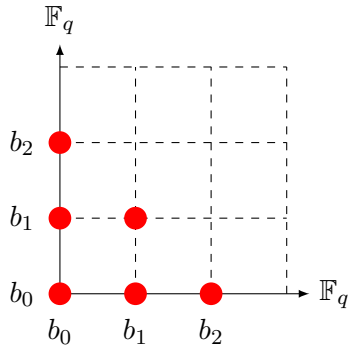
\begin{figure}
\begin{tikzpicture}[scale=1]
\foreach \x in {1,2,...,3}{
\draw [dashed] (0,\x)node[left]{} -- (3,\x)node[right] {};
\draw [dashed] (\x,0)node[below]{} -- (\x,3)node[right] {};}

\draw [-latex] (0,0) -- (3.3,0)node[right] {$\F_q$};
\draw [-latex] (0,0) -- (0,3.3)node[above] {$\F_q$};

\foreach \y in {0}
\foreach \x in {0,...,2} \fill [color=red](\x,\y) {circle(.15cm)};
\foreach \y in {1}
\foreach \x in {0,...,1} \fill [color=red](\x,\y) {circle(.15cm)};
\foreach \y in {2}
\foreach \x in {0} \fill [color=red](\x,\y) {circle(.15cm)};

\foreach \x in {0} \fill [](\x,-0.2)node[below]{$b_0$};
\foreach \x in {1} \fill [](\x,-0.2)node[below]{$b_1$};
\foreach \x in {2} \fill [](\x,-0.2)node[below]{$b_2$};

\foreach \y in {0} \fill [](-0.2,\y)node[left]{$b_0$};
\foreach \y in {1} \fill [](-0.2,\y)node[left]{$b_1$};
\foreach \y in {2} \fill [](-0.2,\y)node[left]{$b_2$};

\end{tikzpicture}
   \caption{Set $B(2,3)$, a $2$-dimensional simplex of side length $3$.}
    \label{26.03.11}
\end{figure}

Similarly, the set $B(3,3)$ can be seen as a tetrahedron in $\F_q^3$ with vertices \[(0,0,0), (b_2,0,0), (0,b_2,0), (0,0,b_2).\] 
\end{example}
In general, by ordering the elements of $B = \{b_0, \dots, b_{l-1} \}$ in $\F_q$ such that $b_0 < \cdots < b_{\ell-1}$ and the rest of the elements in $\F_q$ are greater than $b_{\ell-1}$, then $B(m,s)$ is an $m$-dimensional polytope in $\F_q^m $ with vertices in
$(0,\ldots,0), (b_{s-1},0,\ldots,0), \ldots, (0,\ldots,0, b_{s-1})$.

\begin{remark}
Note that
\[|B(m,s)| = \binom{m + s - 1}{m},
\]
which equals the number of monomials of degree at most $s-1$ in $m$ variables.
\end{remark}

Let $\F_q[{\bm x}]_{\leq \nu} := \F_q[x_1,\ldots,x_m]_{\leq \nu}$ be the set of polynomials in $m$ variables of total degree at most $\nu$. Assume that $\F_q^m = \left\{P_1,\ldots, P_{q^m}\right\}$. A {\bf Reed-Muller} code is denoted and defined by 
\[\RM(m,\nu) := \left\{ \left(f(P_1),\ldots,f(P_{q^m})\right) \ : \ f({\bm x}) \in \F_q[{\bm x}]_{\leq \nu} \right\}.\]

Reed-Muller type codes are defined by the evaluation of elements in $\F_q[{\bm x}]_{\leq \nu}$ over arbitrary sets of points. We now define a Reed-Muller type code that was introduced in \cite{kopparty_high_rate} and is given by the evaluation of polynomials on the elements of a simplex.

\begin{definition}
Let $B(m,s) = \left\{P_1, \ldots, P_n \right\}$ be a simplex. For a nonnegative integer $\nu$, we define a {\bf combinatorial array for polynomials} (CAP) code as
\[
\CAP(B,m,s,\nu) := \left\{ \left(f(P_1),\ldots,f(P_n)\right) \ : \ f({\bm x}) \in \F_q[{\bm x}]_{\leq \nu}\right\}.
\]
\end{definition}

\begin{remark}\label{26.05.01}
Observe that in a CAP code, we can always assume that $s=|B|$. Indeed, given $B = \left\{b_0, b_1, \dots, b_{\ell-1} \right\} \subseteq \F_q$, we can define $B_s := \left\{b_0, b_1, \dots, b_{s-1} \right\}$. Thus, we obtain that
\[\CAP(B,m,s,\nu) = \CAP(B_s,m,s,\nu),\]
where $s = |B_s|$.
\end{remark}

As a consequence of Remark~\ref{26.05.01}, from now on, we assume that $s = |B|=\ell$ and we denote the CAP code $\CAP(B,m,s,\nu)$ by $\CAP(B,m,\nu)$.

\begin{example}
The CAP code $\CAP(B,1,\nu)$ is the Reed-Solomon code $\RS(B,\nu)$.
\end{example}
\begin{example}
Take the set $B=\left\{b_0, b_1, b_2 \right\} \subseteq \F_q$. The combinatorial simplex $B(2)$ appears in Figure~\ref{26.03.11}. The CAP code $\CAP(B,2,1)$ is obtained by evaluating all the bivariate polynomials up to degree one on the points that appear in Figure~\ref{26.03.11}.
\end{example}
We say that two codes $C$ and $C^\prime$ over $\F_q$ are {\bf monomially equivalent} if there is a weight-preserving bijective linear map between $C$ and $C'$. The following examples show that distinct ordered sets $B$ and $B^\prime$ in $\F_q$ of the same size may not generate monomially equivalent CAP codes, even when one is a permutation of the other.
\begin{remark}
By taking the ordered sets $B = \{1, 2, 3, 6\} \subseteq \mathbb{F}_7$ and $B' = \{6, 3, 2, 1\} \subseteq \mathbb{F}_7$, using \cite{sagemath, M2, Ball2020CodingTP}, we can verify that the $6$ nonzero scalar multiples of the polynomial $x_1 + x_2 + x_3 + x_4$ give a weight class of size $6$ at weight $19$ in $\mathrm{CAP}(B, 4, 1)$, while $\mathrm{CAP}(B', 4, 1)$ has no codewords of weight $19$. In particular, the two codes are not monomially equivalent.
\end{remark}

\section{Vanishing Ideal}\label{s:vanishing_ideal}

In this section, we determine the vanishing ideal of an $m$-dimensional combinatorial simplex and focus on its properties, including a universal Gr\"obner basis. These properties are relevant to the construction of the dual of a CAP code. We recommend \cite{cox,villarreal_book_monomial_algebras} for the basic terminology from commutative algebra that we will consider. 

Given a set of $n$ distinct ordered points $\mathcal{X}=\{P_1,\ldots,P_n\}$ in $\F_q^m$, with $n\geq 2$, the {\bf evaluation map} is the $\F_q$-linear map given by 
\[
\begin{array}{lccc}
{\ev} \colon & \F_q[{\bm x}] & \rightarrow & \F_q^{n}\quad \\
& f & \mapsto & f(\cX) := \left(f(P_1),\ldots,f(P_n)\right).
\end{array}
\]

The kernel of ${\ev}$, denoted by $\operatorname{I}(\cX)$ and called the {\bf vanishing ideal} of $\mathcal{X}$, consists of the polynomials of $\F_q[{\bm x}]$ that vanish at all points of $\mathcal{X}$.
\begin{remark}\label{26.03.15}
The evaluation map induces an isomorphism of $\F_q$-linear spaces between the quotient ring $\F_q[{\bm x}]/\operatorname{I}(\cX)$ and $\F_q^n$.
\end{remark}

Let $J$ be an ideal in $\F_q[{\bm x}]$. We define the {\bf affine algebraic variety} $\operatorname{V}(J)$ in $\F_q^m$ by
\[
\operatorname{V}(J) := \left\{ P \in \F_q^m \ : \ f(P) = 0 \text{ for all } f \in J \right\}.
\]

The following result shows how to describe a combinatorial simplex as an affine algebraic variety.
\begin{lemma}\label{26.03.13}
Let $B = \left\{b_0, b_1, \dots, b_{\ell-1} \right\}$ be an ordered subset of $\F_q$. Define the ideal
\begin{equation}\label{JB}
    J_{B(m)} := \left \langle \prod^{j_{1}-1}_{i=0}(x_1 - b_i)  \dots \prod^{j_{m}-1}_{i=0}(x_m - b_i) \ : \ j_{1} + \dots + j_{m} = \ell \right \rangle \subseteq \F_q[{\bm x}].
\end{equation}
Then, we have that $B(m) = \operatorname{V}(J_{B(m)})$.
\end{lemma}
\begin{proof}
By definition,
\[
B(m) = \left\{(b_{i_1}, b_{i_2}, \dots, b_{i_m}) \ : \ i_1+\cdots+i_m <\ell, \ i_j \in \N \right \}.
\]
First, we show that $B(m) \subseteq V(J_{B(m)})$. Let $\mathbf{b} = (b_{i_1}, \dots, b_{i_m})$ be an element of $B(m)$ and $f({\bm x})$ a generator of $J_{B(m)}$, i.e.,
\begin{equation*}
f({\bm x}) \;=\; \prod_{i=0}^{j_{1}-1} (x_1 - b_i)\cdots \prod_{i=0}^{j_{m}-1} (x_m - b_i),
\end{equation*}
for some $j_1,\dots,j_m$ with $j_1+\dots+j_m=\ell$. If ${i_r} \geq j_{r}$ for all $1 \leq r \leq m$, then
\[
i_1 + \cdots + i_m \geq j_{1} + \cdots + j_{m} = \ell,
\]
which contradicts the definition of $\mathbf{b}$. Hence, there exists $1\leq r \leq m$ such that
$i_r < j_r$. For this~$r$, the element $(x_r - b_{i_r})$ divides $f({\bm x})$, and $f(\mathbf{b}) = 0$.
Thus, $B(m) \subseteq V(J_{B(m)})$.

We now check that \(V(J_{B(m)}) \subseteq B(m)\). Suppose \(\mathbf{b} \notin B(m)\). There are two cases: either some component \(b_j \notin B\), or all of the entries \(b_j \in B\). In the former case, \(\mathbf{b}\) does not vanish on the element
\[
\prod_{i=0}^{\ell-1}(x_j-b_i)\in J_{B(m)},
\]
which means \(\mathbf{b}\notin V(J_{B(m)})\). For the latter, write \(\mathbf{b}=(b_{i_1},\dots,b_{i_m})\notin B(m)\). Then
\[
i_1+\cdots+i_m\ge \ell.
\]
Define
\[
g({\bm x})
:=
\prod_{i=0}^{i_1-1}(x_1-b_i)\cdots
\prod_{i=0}^{i_m-1}(x_m-b_i).
\]
Then \(g({\bm x}) \in J_{B(m)}\), since \(i_1+\cdots+i_m\ge \ell\). However,
\[
g({\bm b})\neq 0,
\]
i.e., \(\mathbf{b}\notin V(J_{B(m)})\).
\end{proof}

The following result is known as the Affine $\F_q$-Nullstellensatz.
\begin{proposition}[\hspace{0.1pt}{\cite[Theorem 2.3]{ghorpade}}] \label{26.03.14} Let $J$ be an ideal in $\F_q[{\bm x}]$. The vanishing ideal of the affine algebraic variety $V(J)$ is given by
\begin{equation*}
    \operatorname{I}(\operatorname{V}(J)) = J + \left<x_1^q-x_1, \dots, x_m^q-x_m\right>.
\end{equation*}
\end{proposition}

We can now describe the vanishing ideal of a combinatorial simplex.
\begin{theorem}\label{t:vanishing_jb}
Let $B = \left\{b_0, b_1, \dots, b_{\ell-1} \right\}$ be an ordered subset of $\F_q$. The vanishing ideal of the simplex $B(m)$ is given by
    $$
    \operatorname{I}(B(m)) = J_{B(m)},
    $$
    where $J_{B(m)}$ is defined in Equation~\eqref{JB}.
\end{theorem}
\begin{proof}
By Lemma~\ref{26.03.13} and Proposition~\ref{26.03.14}, we have that
\[
\operatorname{I}(B(m)) = \operatorname{I}(\operatorname{V}(J_{B(m)})) = J_{B(m)} + \left<x_1^q-x_1, \dots, x_m^q-x_m\right>.
\]
  It suffices to show that $\left<x_1^q-x_1, \dots, x_m^q-x_m\right> \subseteq J_{B(m)}$. For $1 \leq r \leq m$, take $j_r=\ell$.
  Then, 
  \begin{equation*}
      \prod_{i=0}^{j_{1}-1} (x_1 - b_i)\cdots \prod_{i=0}^{j_{m}-1} (x_m - b_i) = \prod^{l-1}_{i=0} (x_r-b_i) \in J_{B(m)},
  \end{equation*}
  which clearly divides $x_r^q - x_r$. Hence, $x_r^q - x_r \in J_{B(m)}$ for $1 \leq r \leq m$, finishing the proof.
\end{proof}

We denote the set of monomials in $\F_q[{\bm x}]$ by $\cM$. A {\bf monomial order} $\prec$ on $\cM$ is a total order with the following properties:
\begin{itemize}
    \item The element 1 is the least monomial.
    \item If $M_1 \prec M_2$, then $M M_1 \prec M M_2$, for all $M, M_1, M_2 \in \cM$.
\end{itemize}
Fix a monomial order $\prec$ in $\cM$ and let $f$ be a nonzero element in $\F_q[{\bm x}]_{\leq r}$. The greatest monomial that appears in $f$ with respect to $\prec$, denoted by $\ini_\prec(f)$, is called the {\bf leading monomial} of $f$. Although the initial depends on $\prec$, we will just denote it $\ini(f)$, since for our purposes in this work, the chosen monomial order will not be relevant. Given an ideal $I \subset \F_q[{\bm x}]$, a {\bf Gr\"{o}bner basis} for $I$ is a set
\[\{f_1, \ldots, f_s\} \subseteq I\]
such that for every polynomial $f \in I \setminus \{0\}$, we have that $\ini(f)$ is a multiple of $\ini(f_i)$ for some $i \in \{1, \ldots, s\}$. The concept of Gr\"{o}bner basis was introduced in ~\cite{buchberger}, where the author proved that if $\{f_1, \ldots, f_s\}$ is a Gr\"{o}bner basis for $I$, then $I = \left<f_1, \ldots, f_s\right>$, and that every ideal admits a Gr\"{o}bner basis with respect to a fixed monomial order. We use in the following sections Gr\"{o}bner basis tools to compute the minimum distance and the GHWs of a CAP code.

The {\bf footprint of an ideal} $I \subset \F_q[{\bm x}]$, denoted by $\Delta(I)$, is the set of all monomials in $\F_q[{\bm x}]$ which are not multiples of any $\ini(f)$, for $f \in I$. The {\bf footprint of a set} $\{f_1, \ldots, f_s\} \subset \F_q[{\bm x}],$  denoted by $\Delta(f_1,\ldots,f_s)$, is the set of all monomials which are not multiples of any $\ini(f_i)$, for $1 \leq i \leq s$. We can see that $\{f_1, \ldots, f_s\}$ is a Gr\"{o}bner basis for $I$ if and only if 
$\Delta(I) = \Delta(f_1,\ldots,f_s)$.
\begin{remark}\label{26.03.16}
An important result in Buchberger's thesis~\cite{buchberger} states that the set of classes $\{ M + I \mid M \in \Delta(I) \} \subseteq \F_q[{\bm x}]_{\leq r}/I$ is a basis for $\F_q[{\bm x}]_{\leq r}/I$ as an $\F_q$-vector space.
\end{remark}

By combining Remarks~\ref{26.03.15} and~\ref{26.03.16}, we obtain the following result that gives conditions to determine if a set of generators is a Gr\"{o}bner basis for a vanishing ideal.
\begin{lemma}\label{l:footprint_size_n}
Let $\cX$ be a set of $n$ distinct points in $\F_q^m$. Assume that $\operatorname{I}(\cX) = \left<f_1, \ldots, f_s\right>$ is the vanishing ideal of $\cX$. We have that $\{f_1,\ldots,f_s\}$ is a Gr\"{o}bner basis for $\operatorname{I}(\cX)$ if and only if
\[
|\Delta(f_1,\ldots,f_s)| = n.
\]
\end{lemma}

By definition, a Gr\"{o}bner basis depends on a fixed monomial order $\prec$. If a set $\{f_1, \ldots, f_s\}$ is a Gr\"{o}bner basis for an ideal $I$ and for any monomial order, then the set is called a {\bf universal Gr\"{o}bner basis} for $I$. Next, we claim that the generators described in Equation~\eqref{JB} form a universal Gr\"{o}bner basis for the ideal $J_{B(m)}$.
\begin{theorem}\label{t:grobner_jb}
The generators of the ideal $J_{B(m)}$ shown in Equation~\eqref{JB} form a universal Gr\"{o}bner basis. Moreover, for any monomial order, the footprint of $J_{B(m)}$ is given by
    \begin{equation*}
        \Delta(J_{B(m)}) = \{x_1^{a_1}\cdots x_n^{a_n}  \ : \ a_1+\dots+a_n < \ell\}.
    \end{equation*}
\end{theorem}
\begin{proof}
Let $g_1, \dots, g_n$ be the generators of $J_{B(m)}$ that are shown in Equation~\eqref{JB}. For any monomial order, the leading terms are
\begin{equation*}
\{\ini(g_1), \dots, \ini(g_n) \} = \{x_1^{i_1}  \cdots  x_m^{i_m} \ : \ i_1 + \dots + i_m = \ell  \}.
\end{equation*}
Thus, for any monomial order $\prec$, we have that
\[
\Delta(g_1,\ldots,g_n) = \{x_1^{i_1}  \cdots  x_m^{i_m} \ : \ i_1 + \dots + i_m < \ell \}.
\]
By definition of $B(m)$, we get $|B(m)| = |\Delta(g_1,\ldots,g_n)|$. Therefore, we obtain the result by Lemma~\ref{l:footprint_size_n}. 
\end{proof}
As a consequence of Theorem~\ref{t:grobner_jb}, we see that the footprint of the ideal $J_{B(m)}$ depends only on $|B|$ and $m$.
\begin{example}
Let $B = \{b_0, b_1, b_2, b_3, b_4\}$ be an ordered subset of $\mathbb{F}_9$. The vanishing ideal of the combinatorial simplex $B(2)$ is given by
\[
    J_{B(2)} = \left\langle \prod_{i=0}^{j_1-1}(x - b_i)\prod_{i=0}^{j_2-1}(y - b_i) \ : \ j_1 + j_2 = 5 \right\rangle \subseteq \mathbb{F}_9[x,y],
\]
whose generators form a universal Gr\"obner basis with leading terms
\[
    \{\operatorname{in}(g_1), \dots, \operatorname{in}(g_n)\} = \{x^5, x^4y, x^3y^2, x^2y^3, xy^4, y^5\}.
\]
Consequently, the footprint of $J_{B(2)}$ is given by all monomials of degree at most $4$, i.e.,
\[\Delta(J_{B(2)}) = \{x^iy^j \ : \ i + j \leq 4\},\]
which forms a simplex on the monomial grid; see Figure~\ref{26.04.20}.

\begin{figure}[h!]
\begin{tikzpicture}[scale=0.60, every node/.style={font=\tiny}]
\gridcell{0}{0}{cellempty}{$1$}
\gridcell{1}{0}{cellempty}{$x$}
\gridcell{2}{0}{cellempty}{$x^2$}
\gridcell{3}{0}{cellempty}{$x^3$}
\gridcell{4}{0}{cellempty}{$x^4$}
\gridcell{0}{1}{cellempty}{$y$}
\gridcell{1}{1}{cellempty}{$xy$}
\gridcell{2}{1}{cellempty}{$x^2\!y$}
\gridcell{3}{1}{cellempty}{$x^3\!y$}
\gridcell{0}{2}{cellempty}{$y^2$}
\gridcell{1}{2}{cellempty}{$xy^2$}
\gridcell{2}{2}{cellempty}{$x^2\!y^2$}
\gridcell{0}{3}{cellempty}{$y^3$}
\gridcell{1}{3}{cellempty}{$xy^3$}
\gridcell{0}{4}{cellempty}{$y^4$}
 \end{tikzpicture}
 \caption{The footprint of $J_{B(m)}$ on the monomial grid.}
    \label{26.04.20}
\end{figure}
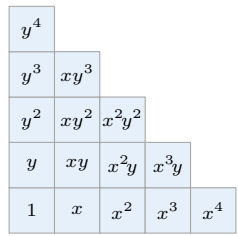

\end{example}

From Theorem \ref{t:grobner_jb}, we can directly obtain the dimension of CAP codes.

\begin{corollary}\label{c:dimension}
    The dimension of the code $\CAP(B,m,\nu)$ is
    \begin{equation*}
       \operatorname{dim}(\CAP(B,m,\nu)) = {m + \nu \choose m}=\binom{m+\nu}{\nu}.
    \end{equation*}
\end{corollary}
\begin{proof}
By Theorem \ref{t:grobner_jb}, any polynomial of degree at most $\nu <\ell$ can be generated by monomials in $\Delta(J_{B(m)})$, whose evaluations are linearly independent by Remark \ref{26.03.16}. The statement follows using a stars and bars argument to compute all the monomials of degree at most $\nu $. 
\end{proof}

\section{Generalized Hamming weights}\label{s:ghws}
In this section, we use Gr\"{o}bner basis tools to find the minimum distance and, more generally, the generalized Hamming weights of a CAP code. Let $f$ be an element in $\F_q[{\bm x}]$. Denote the set of zeros of $f$ in $\cX$ by $V_{\cX}(f)$. Note that the Hamming weight of the element $f(\cX):=\left(f(P_1),\ldots,f(P_n)\right)$ is given by
\[\wH(f(\cX))=n-|V_{\cX}(f)|.\]
Similarly, given any subcode $D\subset \CAP(B,m,\nu)$ with $\dim D=r$, there are $f_1,\dots,f_r$ in $\fq[\xx]_{\leq \nu}$ such that $D=\langle f_1(\cX),\dots,f_r(\cX)\rangle$. Since $\{ f_1(\cX),\dots,f_r(\cX)\}$ is linearly independent, we may assume that $\ini(f_1)\succ \cdots \succ \ini(f_r)$. If we denote $F:=\{f_1,\dots,f_r\} $, then 
\begin{equation}\label{eq:support_variety}
\abs{\supp(D)}=n-\abs{V_{\cX}(F)},
\end{equation}
where $V_{\cX}(F)$ denotes the common zeroes of the polynomials of $F$ in $\cX$. Moreover, we can also see that $V_{\cX}(F)$ is the set of the zeros of the ideal $\operatorname{I}(\cX)+(F) \subset \F_q[{\bm x}]$. The next proposition, known as the footprint bound, allows us to relate the support of the subspace associated with a set of polynomials $F$ with the footprint of a certain ideal. This is the coding-theoretic version of a classical result in algebraic geometry; see \cite[Theorem 6 and Proposition 7, Chapter 5 \textsection 3]{cox}. For any subset $F\subset \fq[\xx]$, we denote $\ini(F):=\{\ini(f) \ : \ f\in F\}$. We denote by $\binom{S}{r}$ the set of all subsets of $S$ with size $r$.

\begin{proposition}\label{p:footprint_bound} 
Let $G=\{g_1,\ldots,g_t\}\subset \operatorname{I}(\cX)$ be a Gröbner basis
for a monomial order $\prec$. For any $F=\{f_1,\dots,f_r\} \subset \F_q[\xx]$, we have that
\[
|V_{\cX}(F)| \leq |\Delta(\ini(G), \ini(F))|.
\]
If $D=\langle f_1(\cX),\dots,f_r(\cX)\rangle$ has dimension $r$, then 
\[\abs{\supp(D)} \geq n-|\Delta(\ini(G), \ini(F))|.\]
\end{proposition}

\begin{remark}
If we apply Proposition \ref{p:footprint_bound} to the CAP code $\CAP(B,m,\nu)$, we obtain

\begin{equation}\label{eq:footprint_bound}
d_r(\CAP(B,m,\nu))\geq \min \left\{\abs{B(m)}- |\Delta(J_{B(m)})\cap \Delta(M)| \ : \ M \in \binom{\cM}{r} \right\}.
\end{equation}
\end{remark}

The following result shows that the footprint bound, i.e., Equation~\eqref{eq:footprint_bound} is sharp. 
\begin{theorem}\label{t:footprint_sharp}
Let $1\leq r \leq \dim  {m + \nu \choose m}$. Then, 
$$
d_r(\CAP(B,m,\nu))=\min \left\{|\Delta(J_{B(m)})\setminus \Delta(M)| \ : \ M \in \binom{\cM}{r} \right\}.
$$
\end{theorem}
\begin{proof}
Since $\abs{\Delta(J_{B(m)})}=\abs{B(m)}$, for a given $M=\{\xx^{\bf \alpha^{(1)}},\dots,\xx^{\alpha^{(r)}}\} \in\binom{\cM}{r}$, we have $\abs{B(m)}- |\Delta(J_{B(m)})\cap \Delta(M)|=|\Delta(J_{B(m)})\setminus \Delta(M)|$. Therefore, the bound from Equation~\eqref{eq:footprint_bound} becomes
$$
d_r(\CAP(B,m,\nu))\geq \min \left\{|\Delta(J_{B(m)})\setminus \Delta(M)| \ : \ M \in\binom{\cM}{r} \right\}.
$$
To prove the reverse inequality, we show that for any monomial set $M=\{\xx^{\bf \alpha^{(1)}}, \dots, \xx^{\alpha^{(r)}}\} \in \binom{\cM}{r}$, we can find a set of polynomials $F=\{f_1,\dots,f_r\}$ such that $\ini(f_i)=\xx^{\alpha^{(i)}}$, for $1\leq i \leq r$, and 
\begin{equation}\label{eq:equality_cardinalities}
|B(m)\setminus V_{B(m)}(F)|=\abs{B(m)}-\abs{V_{B(m)}(F)}=|\Delta(J_{B(m)})\setminus  \Delta(M)|.
\end{equation}
Such a set of polynomials would complete the proof; see Equation~\eqref{eq:support_variety}.

Define the set \[\simplex(m,\ell) := \{ \beta \in \mathbb{N}^m : \sum_{i=1}^m \beta_i<\ell \}. \]  Consider the bijection
\begin{equation}\label{eq:varphi}
\begin{array}{lccc}
\varphi: & \simplex(m,\ell)  & \to & B(m), \\
&(i_1,\dots,i_m)  & \mapsto & (b_{i_1},\dots,b_{i_m}).
\end{array}
\end{equation}

For every $\mathbf{b}=(b_{i_1},\dots,b_{i_m})\in B(m)$, we take the polynomial
$$
f_{\mathbf{b}} := \prod_{i=0}^{j_1-1}(x_1-b_i)\cdots \prod_{i=0}^{j _m-1}(x_m-b_i).
$$
We now define the set $F := \{ f_{\varphi(\alpha^{(i)})},\; 1\leq i \leq r\}$. We claim that $P\in B(m)\setminus V_{B(m)}(F)$ if and only if $\xx^{\varphi^{-1}(P)} \in \Delta(J_{B(m)}) \setminus \Delta(M)$. Let $P\in B(m)$. Note that $P\not\in V_{B(m)}(f_{\varphi(\alpha^{(i)})})$ if and only if $\varphi^{-1}(P)\geq \alpha^{(i)}$, 
where we are considering the partial order in $\mathbb{N}^m$. Since $V_{B(m)}(F)=\bigcap_{i=1}^r V_{B(m)}(f_{\varphi(\alpha^{(i)})})$, then $P\not\in V_{B(m)}(F)$ if and only if, for some $1\leq i \leq r$, we have $\varphi^{-1}(P)\geq \alpha^{(i)}$. By the definition of $\Delta(M)=\Delta(\xx^{\bf \alpha^{(1)}},\dots,\xx^{\alpha^{(r)}})$, we have just proved that $P\in B(m)\setminus V_{B(m)}(F)$ if and only if $\xx^{\varphi^{-1}(P)} \in \Delta(J_{B(m)})\setminus \Delta(M)$
,which implies the equality in Equation~\eqref{eq:equality_cardinalities} and finishes the proof. 
\end{proof}

As a consequence of Theorem~\ref{t:footprint_sharp}, we recover the minimum distance of CAP codes, which was determined in ~\cite{kopparty_high_rate}.

\begin{corollary}\label{min_distance}
We have
$$
d_1(\CAP(B,m,\nu))=\binom{\ell-\nu+m-1}{m}.
$$
\end{corollary}
\begin{proof}
By Theorem \ref{t:footprint_sharp}, we have
$$
d_r(\CAP(B,m,\nu))=\min \left\{|\Delta(J_{B(m)})\setminus \Delta(\xx^\alpha)|:\xx^\alpha \in \cM \right\}.
$$
Since $\xx^\beta \mid \xx^\alpha$ implies $\abs{\Delta(\xx^\beta)}\leq \abs{\Delta(\xx^\alpha)}$, the monomials $\xx^\alpha$ with the lowest value for $|\Delta(J_{B(m)})\setminus \Delta(\xx^\alpha)|$ are those with $\sum_{i=1}^m\alpha_i=\nu$. However, for any such monomial, we have
$$
|\Delta(J_{B(m)})\setminus \Delta(\xx^\alpha)|=\binom{\ell-\nu+m-1}{m}.
$$
Indeed, $\xx^\beta\in \Delta(J_{B(m)})\setminus \Delta(\xx^\alpha)$ if and only if $\beta \geq \alpha$ and $\sum_{i=1}^m\beta_i<\ell$. If we define $\gamma=\beta -\alpha$, this is the same as counting how many tuples $\gamma$ we have with $\gamma \geq 0$ and $\sum_{i=1}^m\gamma_i<\ell -\nu$, which is precisely $\binom{\ell-\nu+m-1}{m}$. 
\end{proof}

\begin{example}
Let $B = \{0, 1, 2\} \subseteq \mathbb{F}_q$ with $q \geq 3$, and consider the code $\mathrm{CAP}(B, 3, 1)$, of length $n = 10$ and dimension $k= 4$. By Theorem~\ref{t:footprint_sharp}, we compute
\[
d_r(\mathrm{CAP}(B, 3, 1)) = \min\left\{ |\Delta(J_{B(3)}) \setminus \Delta(M)| : M \subseteq \{1, x_1, x_2, x_3\},\ |M| = r \right\},
\]
where $\Delta(J_{B(3)}) = \{ x_1^{a_1} x_2^{a_2} x_3^{a_3} : a_1 + a_2 + a_3 < 3 \}$ has $10$ elements. The minimum is attained by the choices in the table below, giving the full generalized Hamming weight hierarchy.
\begin{center}
\begin{tabular}{c|c|c|c}
$r$ & optimal $M$ & $|\Delta(J_{B(3)}) \setminus \Delta(M)|$ & $d_r$ \\ \hline
$1$ & $\{x_1\}$ & $4$ & $4$ \\
$2$ & $\{x_1, x_2\}$ & $7$ & $7$ \\
$3$ & $\{x_1, x_2, x_3\}$ & $9$ & $9$ \\
$4$ & $\{1, x_1, x_2, x_3\}$ & $10$ & $10$
\end{tabular}
\end{center}
Geometrically, an optimal $r$-dimensional subcode is spanned by the evaluations of $r$ of the coordinate functions: for instance, $d_3$ is achieved by $D = \mathrm{span}\{x_1(B(3)), x_2(B(3)), x_3(B(3))\}$, whose support omits only the origin $(0,0,0) \in B(3)$. Note that $d_1 = 4$ also agrees with Corollary~\ref{min_distance}, since $\binom{\ell - \nu + m - 1}{m} = \binom{4}{3} = 4$. These results can be checked using \cite{sanjoseGHWsPackage,githubGHWs}.
\end{example}

\section{Dual codes}\label{s:dual}
In this section, we give an explicit description of the dual of a CAP code. Recall that $B = \left\{b_0, b_1, \dots, b_{\ell-1} \right\}$ is an ordered subset of $\F_q$. We denote by $B^m$ the Cartesian product $B \times \cdots \times B$ with $m$ entries and by $B^c(m)$ the complement of $B(m)$ in $B^m$, i.e., $B^c(m):=B^m\setminus B(m)$.

To describe the dual of a CAP code, we start by finding the vanishing ideal of $B^c(m)$. 
\begin{proposition}\label{p:gb_dual}
Let $B = \left\{b_0, b_1, \dots, b_{\ell-1} \right\}$ be an ordered subset of $\F_q$. A universal Gr\"obner basis for the vanishing ideal $\operatorname{I}(B^c(m))$ is given by
$$
G=\set{\prod^{\ell - 1}_{i = j_1+1}(x_1 -b_i) {\cdots} \prod^{\ell - 1}_{i = j_m+1}(x_m -b_i),  \  \prod^{\ell - 1}_{i=0} (x_j -b_i)  \ : \ j_1 + \dots + j_m = \ell-1, \ j \in [m]}.
$$
The initial ideal of $\operatorname{I}(B^c(m))$, that depends only on $\ell = |B|$ and $m$, is given by
$$
\ini(\operatorname{I}(B^c(m)))=\langle x_1^{\ell-1 -j_1}\cdots x_m^{\ell-1-j_m},\; j_1+\cdots+j_m=\ell-1 \rangle+\langle x_1^{\ell},\dots,x_m^{\ell} \rangle.
$$
\end{proposition}
\begin{proof}
The result follows using the same arguments as in Lemma~\ref{26.03.13} and Theorems~\ref{t:vanishing_jb} and~\ref{t:grobner_jb}. Note that the ideal contains the field equations.
\end{proof}

The next step is to see that a CAP code is the puncturing of a Cartesian code, which we now introduce.
\begin{definition}
Assume that $B^m=\left\{Q_1,\ldots,Q_{\ell^m} \right\}$. The {\bf Cartesian code} of degree $\nu$ is defined by
\[
\Car(B^m,\nu) := \left\{ \left(f(Q_1),\ldots,f(Q_{\ell^m})\right) \ : \ f({\bm x}) \in \F_q[{\bm x}]_{\leq \nu}\right\}.
\]
\end{definition}
Observe that when $m=1$, a Cartesian code is a Reed-Solomon code. When $B = \mathbb{F}_q$, a Cartesian code is a Reed-Muller code. The main parameters of a Cartesian code, length, dimension, and minimum distance, are computed in~\cite{Lopez2014-bq}. The dual of a Cartesian code is given in the following result.
\begin{lemma}\label{26.05.04}
Let $B = \left\{b_0, b_1, \dots, b_{\ell-1} \right\}$ be an ordered subset of $\F_q$ and $l(x) := (x-b_0)\ldots(x-b_{\ell -1})$ the vanishing polynomial of $B$. Denote by $l^\prime(x)$ the formal derivative of $l(x)$ and define the polynomial $L(x_1,\ldots,x_m) := l^\prime(x_1) \cdots l^\prime(x_m)$. The dual of the Cartesian code $\Car(B^m,\nu)$ is given by
\[\Car(B^m,\nu)^\perp = \left\{ \left(L(Q_1)^{-1}f(Q_1),\ldots,L(Q_{\ell^m})^{-1}f(Q_{\ell^m})\right) \ : \ f({\bm x}) \in \F_q[{\bm x}]_{\leq m(\ell-1) - \nu - 1}\right\}.\]
\end{lemma}
\begin{proof}
Let $Q := (b_{j_1},\ldots,b_{j_m})$ be an element in $B^m$. The polynomial \[F_Q(x_1,\ldots,x_m) := \frac{l(x_1)}{(x_1-b_{j_1})} \cdots \frac{l(x_m)}{(x_m-b_{j_m})}\]
has the property that $F_Q(Q_i) \neq 0$ if $Q = Q_i$ and $F_Q(Q_i) = 0$ otherwise. We can also see that $F_Q(Q) = L(Q)$. Thus, the result follows from~\cite[Theorem 2.3]{LOPEZ201913} or \cite[Theorem 5.4]{Lopez2021-tu} by noticing that $\frac{F_Q(x_1,\ldots,x_m)}{L(Q)}$ is the standard indicator function of $Q$ in $B^m$. 
\end{proof}
Let $C \subseteq \F_q^n$ be a code. If we denote by $c \star c^\prime$ the component-wise product between two elements $c$ and $c^\prime$ of $\mathbb{F}_q^n$, and by $c^\prime \star C$ the code $\{ c \star c^\prime : c^\prime \in C \}$, we obtain from Lemma~\ref{26.05.04}
\begin{equation}\label{26.05.05}
\Car(B^m,\nu)^\perp = \lambda_{B^m} \star \Car(B^m, \nu^\perp),
\end{equation}
where $\nu^\perp := m(\ell-1) - \nu - 1$ and $\lambda_{B^m} := (L(Q_1)^{-1}, \ldots, L(Q_{\ell^m})^{-1})$.

The following is a classical result in coding theory that states that the dual of the shortening is the puncturing of the dual.
\begin{lemma}[\hspace{0.1pt}{\cite[Theorem 1.5.7]{Huffman_Pless_2003}}] \label{26.05.03} 
For a set $S \subset [n]$, we have
\[(C^\perp)_S = (C^S)^\perp \qquad \text{ and } \qquad (C^\perp)^S = (C_S)^\perp.\]
\end{lemma}

By definition, we have that $B(m) \subset B^m$. Therefore, we obtain that a CAP code is the puncturing of a Cartesian code. Specifically, we have
\begin{equation}\label{26.05.02}
\CAP(B,m,\nu) = \Car(B^m,\nu)^{B^c(m)}.
\end{equation}

Combining Equation~\eqref{26.05.05}, Equation~\eqref{26.05.02}, and Lemma~\ref{26.05.03}, we obtain
\begin{align*}
\CAP(B,m,\nu)^\perp
&= \left(\Car(B^m,\nu)^{B^c(m)}\right)^\perp\\
&= \left(\Car(B^m,\nu)^\perp\right)_{B^c(m)}\\
&= \left(\lambda_{B^m} \star \Car(B^m, \nu^\perp)\right)_{B^c(m)}\\
&= \lambda_{B(m)} \star \Car(B^m, \nu^\perp)_{B^c(m)},
\end{align*}
where $\lambda_{B(m)}$ is the restriction of $\lambda_{B^m}$ to the entries indexed by the elements in $B(m)$.

We now have the tools to describe the dual of a CAP code.
\begin{theorem}\label{t:dual}
Let $G=\{g_1,\dots,g_t\}$ be the universal Gr\"obner basis of the ideal $\operatorname{I}(B^c(m))$ described in Proposition \ref{p:gb_dual}. Define the set 
$$
\Gamma := \set{M g_i \ : \ M \in \cM, \ \deg (M g_i) \leq \nu^\perp, \ \ini(M g_i)\not \in \langle \ini(g_1),\dots,\ini(g_{i-1}), \ x_1^\ell,\dots,x_m^\ell \rangle }. 
$$
The dual of the CAP code $\CAP(B,m,\nu)$ is given by
\begin{equation*}
\CAP^\perp(B,m,\nu) = \lambda_{B(m)} \star
\operatorname{span}_{\F_q} \left\{ f(B(m))  \ : \ f \in \Gamma \right\}.
\end{equation*}
In other words, up to multiplication by $\lambda_{B(m)}$, the evaluation of the polynomials in $\Gamma$ is a basis for the dual code. 
\end{theorem}
\begin{proof}
By the discussion after Lemma~\ref{26.05.03}, we have that
\begin{align*}
\CAP(B,m,\nu)^\perp
&= \lambda_{B(m)} \star \Car(B^m, \nu^\perp)_{B^c(m)}\\
&= \lambda_{B(m)} \star \operatorname{span}_{\F_q} \left\{f(B(m)) \ : \ f \in \fq[\xx]_{\leq \nu^\perp}, \ f(B^c(m))= \mathbf{0} \right\}.
\end{align*}
We have that $f(B^c(m)) = \mathbf{0}$ and $\deg(f) \leq \nu^\perp $ if and only if $f\in \operatorname{I}(B^c(m))_{\leq \nu^\perp}$. By construction, $\Gamma\subset \operatorname{I}(B^c(m))_{\leq \nu^\perp}$. 

By Proposition~\ref{p:gb_dual}, we know
$$
\ini(\operatorname{I}(B^c(m)))=\langle x_1^{\ell-1 -j_1}\cdots x_m^{\ell-1-j_m},\; j_1+\cdots+j_m=\ell-1 \rangle+\langle x_1^{\ell},\dots,x_m^{\ell} \rangle.
$$
Now extend the map $\varphi$ from Equation~\eqref{eq:varphi} to $[0,\ell-1]^m$ instead of $\simplex(m,\ell)$. Thus,
\begin{eqnarray}
&&\varphi(\set{\alpha \in \mathbb{N}^m \ : \ \xx^\alpha \in \ini(\operatorname{I}(B^c(m)))}\cap B^m \nonumber \\
&&=\varphi\left(\set{\alpha \in \mathbb{N}^m \ : \ \ell-1-j_i\leq \alpha_i \leq \ell-1 
\text{ for some } j_1,\dots,j_m \text{ with } \sum_{i=1}^m j_i=\ell-1} \right). \label{26.05.06}
\end{eqnarray}
As $\alpha_i + j_i \geq \ell -1$ for $i \in [m]$ and $\sum_{i=1}^m j_i=\ell-1$, then $\sum_{i=1}^m \alpha_i \geq (m-1)(\ell -1)$. Similarly, if $\sum_{i=1}^m \alpha_i \geq (m-1)(\ell -1)$ and $\alpha_i \leq \ell - 1$, then $\alpha_i\geq  \ell-1-j_i$ for some $j_1,\dots,j_m$ with $\sum_{i=1}^m j_i=\ell-1$. Thus, from Equation~\eqref{26.05.06}, we obtain 
$$\varphi( \set{ \alpha \in \mathbb{N}^m \ : \ \xx^\alpha \in \ini(\operatorname{I}(B^c(m)))}\cap B^m  = B^m \setminus B\left( m, (m-1)(\ell-1) \right),$$
where we are extending the definition of $B(m,s)$ from Definition \ref{26.03.12} to the case $s>\ell$ (note that this is no longer a simplex). Therefore, by the definitions of $G$ and $\Gamma$, we have the following (see Figure~\ref{fig:cube} for a visual representation of these decompositions):
\begin{align*}
&\abs{ \varphi(\set{ \alpha \in \mathbb{N}^m \ : \ \xx^\alpha \in \ini(\Gamma)})}\\
&=\abs{\varphi\left(\{\alpha \in \mathbb{N}^m \ : \ \xx^\alpha \in \ini(\operatorname{I}(B^c(m)))\}\right)\cap \left[B^m \setminus \varphi(\set{\alpha \ : \ \xx^\alpha \in \fq[\xx]_{> \nu^\perp}})\right] }\\
&=\abs{B^m\setminus \left[ B(m,(m-1)(\ell-1)) \sqcup B^c(m,m(\ell-1)-\nu) \right]}\\
&=\abs{B^m}-\abs{B(m,(m-1)(\ell-1))}-\abs{B^c(m,m(\ell-1)-\nu)},
\end{align*}
where $B^c(m,m(\ell-1)-\nu)=B^m\setminus B(m,m(\ell-1)-\nu)$. Note that $\abs{\ini(\Gamma)}=\abs{\Gamma}$. We also have
\begin{align*}
&\abs{B^m}-\abs{B(m,(m-1)(\ell-1))}= \abs{B(m)} = n \text{ and }\\
&\abs{B^c(m,m(\ell-1)-\nu)}=\abs{B^m}-\abs{B(m,m(\ell-1)-\nu)}=\abs{B(m,\nu+1)}=\dim \CAP(B,m,\nu),
\end{align*}
where we have used Corollary \ref{c:dimension}. In other words, $\abs{\Gamma}=n-\dim \CAP(B,m,\nu)=\dim \CAP^\perp(B,m,\nu)$.

Furthermore, the set $\set{f(B(m)) \ : \ f \in \Gamma }$ is linearly independent. Indeed, if there is a linear combination $f^* := \sum_{f\in \Gamma} \lambda_ff$ that vanishes at all the points of $B(m)$, then $f^*$ also vanishes at all the points of $B^c(m)$ since $\Gamma \subset I(B^c(m,\ell))$. Therefore, 
$$
f^* \in \left \langle \prod_{i=0}^{\ell-1}(x_1-b_i),\dots, \prod_{i=0}^{\ell-1}(x_m-b_i)\right\rangle,
$$
which is not possible because it would would imply that $\ini(f^*)\in \langle x_1^{\ell}, \dots, x_m^{\ell} \rangle$. 
\end{proof}

\begin{figure}[H]
    \centering
    \tdplotsetmaincoords{65}{115} 

\begin{tikzpicture}[tdplot_main_coords, scale=1.0, line join=round, line cap=round]
    
    \def\L{5}      
    \def\nval{1.5} 
    
    \draw[thick,->] (0,0,0) -- (\L+\nval,0,0) node[anchor=north east]{$x_1$};
    \draw[thick,->] (0,0,0) -- (0,\L+\nval,0) node[anchor=north west]{$x_2$};
    \draw[thick,->] (0,0,0) -- (0,0,\L+\nval) node[anchor=south]{$x_3$};
    
    \draw[dashed, gray] (0,0,0) -- (\L,0,0);
    \draw[dashed, gray] (0,0,0) -- (0,\L,0);
    \draw[dashed, gray] (0,0,0) -- (0,0,\L);
    
    \filldraw[fill=blue!30, opacity=0.4, draw=blue, thick] 
        (\L,0,0) -- (0,\L,0) -- (0,0,\L) -- cycle;
        
    \filldraw[fill=red!40, opacity=0.6, draw=red!80!black, thick] 
        (\nval,0,0) -- (0,\nval,0) -- (0,0,\nval) -- cycle;

    \draw[thick, gray!80!black] (\L,0,0) -- (\L,\L,0) -- (0,\L,0);
    \draw[thick, gray!80!black] (\L,0,0) -- (\L,0,\L);
    \draw[thick, gray!80!black] (0,\L,0) -- (0,\L,\L);
    \draw[thick, gray!80!black] (0,0,\L) -- (\L,0,\L) -- (\L,\L,\L) -- (0,\L,\L) -- cycle;
    \draw[thick, gray!80!black] (\L,\L,0) -- (\L,\L,\L);

    \filldraw[fill=blue!30, opacity=0.4, draw=blue, thick] 
        (\L,\L,0) -- (\L,0,\L) -- (0,\L,\L) -- cycle;

    \filldraw[fill=red!40, opacity=0.6, draw=red!80!black, thick] 
        (\L,\L,\L-\nval) -- (\L,\L-\nval,\L) -- (\L-\nval,\L,\L) -- cycle;

    \node[anchor=north east, font=\small, xshift=0mm, yshift=2mm] at (\nval,0,0) {$\nu$};
    \node[anchor=north east, font=\small, xshift=0mm, yshift=3mm] at (\L,0,0) {$\ell-1$};
    
    \node[anchor=north, font=\small, xshift=2mm, yshift=3mm] at (0,\nval,0) {$\nu$};
    \node[anchor=north west, font=\small, xshift=0mm, yshift=4mm] at (0,\L,0) {$\ell-1$};
    
    \node[anchor=east, font=\small] at (0,0,\nval) {$\nu$};
    \node[anchor=east, font=\small] at (0,0,\L) {$\ell-1$};
    
    \fill[black] (\L,\L,\L) circle (1.5pt);
    \node[anchor=west, font=\small] at (\L,\L,\L) {$(\ell-1, \ell-1, \ell-1)$};

\end{tikzpicture} 
    \caption{Regions considered in the proof of Theorem \ref{t:dual}.}
    \label{fig:cube}
\end{figure}
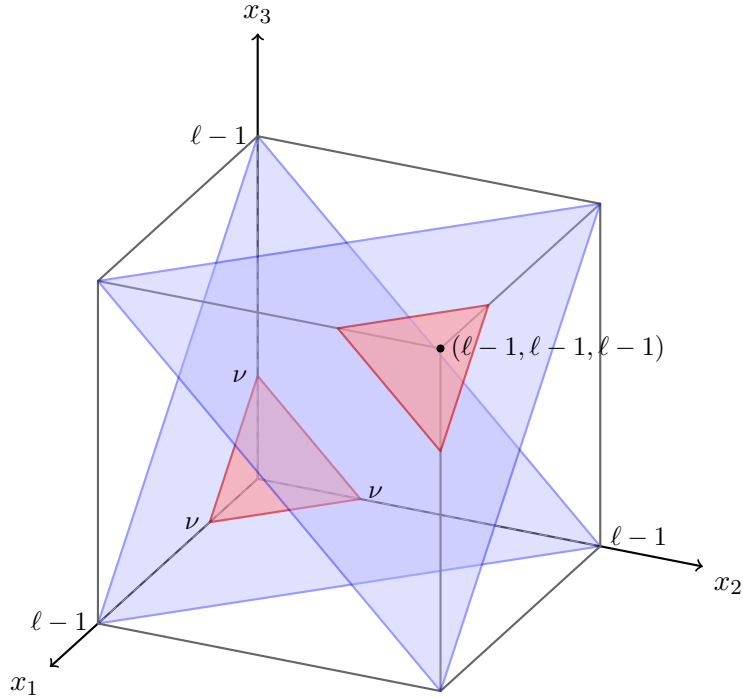

\section{Permutation group}\label{s:permutation}
We continue with the same notation from previous sections. In particular, $B = \left\{b_0, b_1, \dots, b_{\ell-1} \right\}$ is an ordered subset of $\F_q$, $B(m)$ denotes the combinatorial simplex, and $\CAP(B,m,\nu)$ is the CAP code of degree $\nu$. In this section, we focus on the permutation group of $\CAP(B,m,\nu)$. We prove that when $\nu < \ell/2$, its permutation group is given only for the affine transformations that leave invariant the combinatorial simplex $B(m)$. In addition, we describe such affine transformations.

Any permutation $\pi \in S_n$ defines the map
\[\begin{array}{lccc}
& \F_{q}^n & \to & \F_{q}^n \\
& {a}=(a_1,\ldots, a_n) & \mapsto & \pi({a}):=\left(a_{\pi(1)},\ldots, a_{\pi(n)}\right),
\end{array}\]
which is a permutation of the entries of ${a}$.

\begin{definition}
Let $C \subseteq \F_q^n$ be a linear code. For an element $\pi$ of the symmetric group $S_n$, we define
\[\pi(C) := \{\pi(c) : c \in C \}.\]
The {\bf permutation group} of $C$ is the subgroup of the symmetric group $S_n$ defined by
\begin{equation*}
\Per(C) := \left\lbrace\pi\in S_n: \pi(C) = C \right\rbrace.
\end{equation*}
\end{definition}
The permutation group tells us which coordinates of every element $c\in C$ we can permute and still get an element of the code $C$. If $G$ is the generator matrix of a code $C$, the permutation group asks for the columns we can permute in $G$ and still get a generator matrix of the code $C$.

Assume that $B(m)=\left\{P_1,\ldots,P_n\right\}$. For every permutation $\pi$ in $\Per(\CAP(B,m,\nu))$ and every element $f(B(m)) = \left(f(P_1),\ldots,f(P_n)\right) $ in $\CAP(B,m,\nu)$, we have that 
\[\left(\pi \circ f\right)(B(m)) := \pi\left(f(P_1),\ldots,f(P_n)\right) = \left(f(P_{\pi(1)},\ldots,f(P_{\pi(n)}\right) \in \CAP(B,m,\nu).\]
Thus, we see that the element $\pi \in \Per(\CAP(B,m,\nu))$ defines a permutation of $B(m)$, which is a function $T_{\pi} \colon \F_q^m \to \F_q^m$ such that
\[T_{\pi} (P_i) = P_{\pi(i)} \quad \text{ for } i=1,\ldots, n.\]
As $T_{\pi} \colon \F_q^m \to \F_q^m$, and $P_1,\ldots,P_n$ are $n$ different points in $\F_q^m$, there are $m$ polynomials $T_1, \ldots, T_m \colon \F_q^m \to \F_q$ such that
\[T_{\pi}({\bm x}) = \left(T_1({\bm x}),\ldots,T_m({\bm x})\right).\]
\begin{remark}\label{26.05.11}
By Theorem~\ref{t:vanishing_jb}, we know that $\operatorname{I}(B(m)) = J_{B(m)}$. Using Theorem~\ref{t:grobner_jb}, there is an element $T_i^\prime({\bm x})$ for $i=1,\ldots,m$ such that $T_i(B(m)) = T_i^\prime(B(m))$ and $\deg(T_i^\prime) < \ell$. Thus, from now on, we assume that $\deg(T_i) < \ell$. In other words, we assume that every $T_i({\bm x})$ is an $\F_q$-combination of monomials of $\Delta(J_{B(m)}) = \{x_1^{a_1}\cdots x_n^{a_n}  \ : \ a_1+\dots+a_n < \ell\}$.
\end{remark}
\begin{remark}\label{26.05.09}
If we define $T_\pi(B(m))$ as the ordered set $\left\{T_\pi(P_1), \ldots, T_\pi(P_n) \right\}$, then we have that
\[ \left(\pi \circ f\right)(B(m)) = \left(f \circ T_\pi\right)(B(m)).\]
\end{remark}

\begin{definition}
We say that a permutation $\pi$ in $\Per(\CAP(B,m,\nu))$ is an {\bf affine permutation} if $T_{\pi}$ is an affine transformation. In other words, there is an $m \times m$ matrix $A_\pi$ with entries in $\F_q$ and an element $b_\pi \in \F_q^m$ such that \[T_{\pi}({\bm x}) = A_{\pi}{\bm x} +b_\pi.\]
\end{definition}
\begin{lemma}\label{26.05.12}
A permutation $\pi$ in $\Per(\CAP(B,m,\nu))$ is an affine permutation if and only if $\deg(T_i) \leq 1$ for every $T_i$ in $T_{\pi}({\bm x}) = \left(T_1({\bm x}),\ldots,T_m({\bm x})\right).$
\end{lemma}
\begin{proof}
Assume that $T_{\pi}({\bm x}) = A_{\pi}{\bm x} +b_\pi$ is affine. Denote $i$-th row of $A_{\pi}$ by $A_i$. Then, $T_i({\bm x}) = A_i \cdot {\bm x} + b_i$, where $A_i \cdot x$ denotes the standard inner product. Thus, $\deg(T_i) \leq 1$. The converse is true by noticing that if $\deg(T_i) \leq 1$, there is an element $A_i \in \F_q^m$ and $b_i \in \F_q$ such that $T_i({\bm x}) = A_i \cdot {\bm x} + b_i$.
\end{proof}
The group of affine transformations that permute the points of the combinatorial simplex $B(m)$ is denoted by $\operatorname{GA}(B, m)$. 

From Remark~\ref{26.05.11}, we obtain that every element $\pi$ in $\Per(\CAP(B,m,\nu))$ defines a unique element $T_{\pi}({\bm x})~\colon~\F_q^m~\to~\F_q^m$ that permutes the points of $B(m)$. The converse is not necessarily true. However, if the function $T_{\pi}({\bm x})$ is an affine transformation, i.e., if $T_{\pi}({\bm x}) \in \operatorname{GA}(B, m)$, then it is true, as the following result shows.
\begin{proposition}
For every $T({\bm x})$ in $\operatorname{GA}(B, m)$, there is $\pi$ in $\Per(\CAP(B,m,\nu))$ such that $T_{\pi}({\bm x}) = T({\bm x})$.
\end{proposition}
\begin{proof}
Let $f({\bm x})$ be an element in $\F_q[{\bm x}]_{\leq \nu}$. As $T({\bm x})$ is an affine transformation, following the the proof of Lemma~\ref{26.05.12}, 
$T({\bm x}) = \left(T_1({\bm x}),\ldots,T_m({\bm x})\right)$ with $\deg(T_i) = 1$. We obtain that
\[
(f \circ T)({\bm x}) = f(T({\bm x})) = f(T_1({\bm x}),\ldots,T_m({\bm x})).
\]
As $\deg(T_i) \leq 1$, then $\deg(f) = \deg(f \circ T)$. Thus, the element 
$(f \circ T)(B(m))$ is also in $\CAP(B,m,\nu)$. By defining $\pi \in S_n$ such that $\pi(i):=j$ if $T(P_i) = P_j$, we obtain from Remark~\ref{26.05.09} that
\begin{align*}
\left(\pi \circ f\right)(B(m))
&= \left(f \circ T\right)(B(m))\\
&= \left(f(T(P_1)),\ldots,f(T(P_n))\right).
\end{align*}
Therefore, $T_{\pi}({\bm x}) = T({\bm x})$, which completes the proof.
\end{proof}

\begin{remark}\label{26.05.13}
As every element $T({\bm x}) \in \operatorname{GA}(B, m)$ defines a permutation $\pi \in \Per(\CAP(B,m,\nu))$, we write
\[\operatorname{GA}(B, m) \subseteq \Per(\CAP(B,m,\nu)).\]
\end{remark}
The following theorem shows that when $\nu < \ell/2$, the permutation group is affine.
\begin{theorem}\label{26.05.17}
If $0 < \nu < \ell/2$, then
\[
\Per(\CAP(B,m,\nu)) = \operatorname{GA}(B, m).
\]
In other words, the permutation group is given by the affine transformations that permute the elements of $B(m)$.
\end{theorem}
\begin{proof}
From Remark~\ref{26.05.13}, we just need to prove that $\Per(\CAP(B,m,\nu)) \subseteq \operatorname{GA}(B, m)$. Let $\pi$ be an element in $\Per\left(\CAP(B,m,\nu)\right)$ and $T_{\pi}({\bm x}) = \left(T_1({\bm x}),\ldots,T_m({\bm x})\right)$ the permutation of $B(m)$ defined by $\pi$. For $1 \leq j \leq \nu$, we use induction to prove that $\deg(T_i^j) \leq \nu$.

\begin{itemize}
\item[$\bullet$] $(j=1)$ By Remark~\ref{26.05.09}, we have that
\begin{equation}\label{26.05.10}
\left(\pi \circ f\right)(B(m)) = \left(f \circ T_\pi\right)(B(m)) \in \CAP(B, m, \nu)
\end{equation}
for every element $f({\bm x}) \in \F_q[{\bm x}]_{\leq \nu}$. As $0 < \nu$, the polynomials $f_i({\bm x}) := x_i$ are in $\F_q[{\bm x}]_{\leq \nu}$ for $i = 1,\ldots, m$. By Equation~\eqref{26.05.10}, we get
\begin{align*}
\left(\pi \circ f_i\right)(B(m))
&= \left(f_i \circ T_\pi\right)(B(m))\\
&= T_i(B(m)) \in \CAP(B, m, \nu).
\end{align*}
Thus, for $i = 1,\ldots, m$, there is an element $G_i \in \F_q[{\bm x}]_{\leq \nu}$ such that $T_i(B(m)) = G_i(B(m))$. By Remark~\ref{26.05.11}, we know that $\deg(T_i) < \ell$. By Theorem~\ref{t:vanishing_jb}, the evaluation of a polynomial of degree less than $\ell$ is unique, thus, we get that $T_i({\bm x}) = G_i({\bm x}) \in \F_q[{\bm x}]_{\leq \nu}$ and $\deg(T_i) \leq \nu$.
\item[$\bullet$] $(1 < j \leq \nu)$ Assume that for $j^\prime < j$, $\deg(T_i^{j^\prime}) \leq \nu$. Thus,
\begin{align*}
\deg(T_i^j) &= \deg(T_i) + \deg(T_i^{j-1})\\
&\leq \nu + \nu < \ell/2 + \ell/2 = \ell.
\end{align*}
For $i = 1,\ldots, m$, the polynomials $f_i({\bm x}) := x_i^j $ belong to $\F_q[{\bm x}]_{\leq \nu}$. By Equation~\eqref{26.05.10}, we obtain
\begin{align*}
\left(\pi \circ f_i\right)(B(m))
&= \left(f_i \circ T_\pi\right)(B(m))\\
&= T_i^j(B(m)) \in \CAP(B, m, \nu).
\end{align*}
Then, for $i = 1,\ldots, m$, there is an element $G_i \in \F_q[{\bm x}]_{\leq \nu}$ such that $T_i^j(B(m)) = G_i(B(m))$. As $\deg(T_i^j) < \ell$ and the evaluation of a polynomial of degree less than $\ell$ is unique by Theorem~\ref{t:vanishing_jb}, we get $T_i^j({\bm x}) = G_i({\bm x}) \in \F_q[{\bm x}]_{\leq \nu}$ and $\deg(T_i^j) \leq \nu$.
\end{itemize}
By taking $j=\nu$, we see that $\deg(T_i^\nu) \leq \nu$, which means that $\deg(T_i) \leq 1$ for $i=1,\ldots,m$. By Lemma~\ref{26.05.12}, we obtain the result.
\end{proof}

\subsection{Affine transformations that permute}
We now focus on $\operatorname{GA}(B, m)$, the affine transformations that permute the points of $B(m)$.

Let $H$ be a hyperplane in $\F_q^m$. We say that $H$ is {\bf maximal} over $B(m)$ if $|B(m) \cap H|$ is maximal, i.e., $H$ intersects $B(m)$ on the maximum number of points. 

\begin{lemma}\label{26.05.07}
Assume that $b_0 = 0$ in $B = \{b_0, \dots, b_{\ell - 1}\}$ and let $H$ be a hyperplane over $\F_q^m$. If $H$ is maximal over $B(m)$, then \[|H \cap B(m)| =\binom{\ell + m - 2}{m - 1}.\]
\end{lemma}
\begin{proof}
As $b_0=0$, each coordinate hyperplane $H_i : x_i = 0$ intersects $B(m)$ in a smaller simplex $B(m-1)$, which has $\binom{\ell + m - 2}{m - 1}$ points.

Suppose that a hyperplane $H: \sum_{i=1}^m a_i x_i = b$ intersects $B(m)$, and call the intersection $S_H$. As $H$ is a hyperplane, there is $j$ such that $a_j \neq 0$. Project $S_H$ onto $H_j$ by dropping the $j$-th coordinate. This projection is injective because, given the other coordinates, $x_j$ is uniquely determined by the equation of $H$. As $b_0 = 0$, the image of the projection lies in $B(m)$; thus
\begin{equation*}
|H \cap B(m)| \leq |H_j \cap B(m)| = \binom{\ell + m - 2}{m - 1},
\end{equation*}
which completes the proof.
\end{proof}

\begin{proposition}\label{hyperlanes lemma}
Assume that $b_0 = 0$ in $B = \{b_0, \dots, b_{\ell - 1}\}$. The only maximal hyperplanes over $B(m)$ are $H_i : x_i = 0$ for $1 \leq i \leq m$ and, possibly, $H_0 : x_1 + \dots + x_m = b_{\ell - 1}$.    
\end{proposition}
\begin{proof}
Let $H: \sum_{i=1}^m a_i x_i = b$ be a maximal hyperplane over $B(m)$. By Lemma~\ref{26.05.07},
\[|H \cap B(m)| \leq |H_j \cap B(m)|.\]
Assume that equality holds and that $a_j\neq 0$, for some $1\leq j \leq m$. Since both sets are finite of the same size, the injectivity described in the proof of Lemma~\ref{26.05.07} upgrades to a bijection, so every point of $H_j \cap B(m)$ has a unique preimage in $S_H = H \cap B(m)$. For the $i$-th standard vector $e_i$ in $\F_q^m$, with $i \neq j$, we have that $b_{\ell -1}e_i \in H_j \cap B(m)$. Its preimage in $S_H$ must also be $b_{\ell -1}e_i$, since this is the only lattice point of $B(m)$ whose projection is $b_{\ell - 1}e_i$. Plugging into the equation of $H$, we get
\begin{equation*}
a_ib_{\ell - 1}= b \quad \Longrightarrow \quad a_i = \dfrac{b}{b_{\ell - 1}} \qquad \text{for all } i \neq j.
\end{equation*}
If $b = 0$, then $a_i = 0$ for all $i \neq j$, and the equation of $H$ becomes $a_j x_j = 0$, i.e. $H = H_j$. If $b \neq 0$, we obtain
\[H : \sum_{i \neq j}\frac{b}{b_{\ell -1}}x_i + a_jx_j = b.\]
Equivalently, we can rescale $H$ to obtain
\begin{equation}\label{eq:H-rescaled}
H: \sum_{i \neq j} x_i + \dfrac{a_jb_{\ell - 1}}{b} x_j = b_{\ell - 1}. 
\end{equation}
Since we can assume that $H$ is different from $H_i$ for all $i$, then there must be $j'\neq j$ with $a_{j'}\neq 0$. By the same argument, we obtain
\begin{equation}\label{eq:H-rescaled2}
H: \sum_{i \neq j'} x_i + \dfrac{a_j'b_{\ell - 1}}{b} x_j' = b_{\ell - 1}. 
\end{equation}
As Equations~\eqref{eq:H-rescaled} and \eqref{eq:H-rescaled2} represent the same hyperplane, all the coefficients of $x_i$'s should be $1$, i.e., 
\begin{equation*}
H = H_0 : x_1 + \dots + x_m = b_{\ell - 1}. 
\end{equation*}
In this case, we see that all the diagonal points $\{ (b_{i_1},\dots,b_{i_m}) : i_1+\ldots+i_m = \ell-1 \}$ of $B(m)$ form the hyperplane $H_0$.

Therefore, joining the cases $b=0$ and $b\neq0$, we see that any hyperplane besides $ H_0, H_1, \dots, H_m$ intersects $B(m)$ in strictly fewer than $\binom{\ell + m - 2}{m - 1}$ points.
\end{proof}

The following result shows that the affine transformations that permute the points of $B(m)$ also permute the vertices of $B(m)$.
\begin{lemma}\label{26.05.15}
Assume that $b_0 = 0$ in $B = \{b_0, \dots, b_{\ell-1}\}$ and define the vertices $v_0 := \mathbf{0}$ and $v_i := b_{\ell-1}e_i$ for $i=1,\ldots, m$. If $T({\bm x}) \in \operatorname{GA}(B,m)$, then $T({\bm x})$ permutes the vertices. In other words, $T(v_i) = v_j$ for $i,j \in \{0,\ldots,m\}$.
\end{lemma}
\begin{proof}
Let $T({\bm x}) = A{\bm x} + b$ be an element in $\operatorname{GA}(B,m)$. By Lemma~\ref{hyperlanes lemma}, the only maximal hyperplanes over $B(m)$ are $H_i : x_i = 0$ for $1 \leq i \leq m$ and, possibly, $H_0 : x_1 + \dots + x_m = b_{\ell - 1}$.

Assume that $H_0,\ldots, H_m$ are the maximal hyperplanes over $B(m)$. Note that
\begin{equation}\label{26.05.14}
v_i = \bigcap_{j \neq i} H_j \qquad \text{ for } i=0,\ldots,m.
\end{equation}
As $T({\bm x})$ is a bijection,
\[T(v_i) = \bigcap_{j \neq i} T(H_j) \qquad \text{ for } i=0,\ldots,m.\]
The image $T(H_j)$ is also a maximal plane over $B(m)$ because $T({\bm x})$ is affine and a bijection. Therefore, by Equation~\ref{26.05.14}, $T(v_i)$ is also a vertex.

Now assume that $H_1,\ldots, H_m$ are the only maximal hyperplanes over $B(m)$. From the previous paragraph, $T(v_0) = v_0$. Define the lines
\[L_i := \bigcap_{j \neq i} H_j \qquad \text{ for } i=1,\ldots,m.\]
Note that $L_i = \left\{b_0e_i,\ldots,b_{\ell-1}e_i \right\}$, where $e_i$ is the $i$-th standard vector. As $T({\bm x})$ is affine and a bijection, then
\[T(L_i) = \bigcap_{j \neq i} T(H_j) = L_{\sigma(i)} \qquad \text{ for } i=1,\ldots,m \text{ and a permutation } \sigma \in S_m.\]
There are $m$ elements $b_{j_1},\ldots,b_{j_m}$ in $B$ such that
\[T(b_{j_i}e_i) = b_{\ell -1}e_{\sigma(i)} \qquad \text{ for } i=1,\ldots,m.\]
If $j_i = \ell -1$ for all $i=1,\ldots,m$, the proof has finished. Otherwise, there is $i$ such that $\ell - 1 - j_i > 0$. For $i^\prime \neq i$, $T(b_{\ell - 1 - j_i}e_{i^\prime}) = b_j e_{\sigma(i^\prime)}$, for some $j$. Note that $j \neq 0$ because $T(v_0) = v_0$. Therefore, $b_{j_i}e_i + b_{\ell - 1 - j_i}e_{i^\prime}$ is an element of $B(m)$. However, $T(b_{j_i}e_i + b_{\ell - 1 - j_i}e_{i^\prime}) = b_{\ell -1}e_{\sigma(i)} + b_j e_{\sigma(i^\prime)}$ ($T$ is linear since $T(0)=0)$, which is not in $B(m)$.
\end{proof}
\begin{lemma}\label{26.05.16}
Assume that $b_0 = 0$ in $B = \{b_0, \dots, b_{\ell-1}\}$ and define the vertices $v_0 := \mathbf{0}$ and $v_i := b_{\ell-1}e_i$ for $i=1,\ldots, m$. Any element $T({\bm x}) \in \operatorname{GA}(B,m)$ is fully determined by the values $T(v_i)$ for $i \in \{0,\ldots,m\}$.
\end{lemma}
\begin{proof}
Let $T({\bm x}) = A{\bm x} + b$ be an element in $\operatorname{GA}(B,m)$. We have $T(v_0) = b$ and $T(v_i) = A b_{\ell-1}e_i + b = b_{\ell-1} A^i +b$, where $A^i$ is the $i$-th column of $A$. Thus, we obtain that $b = T(v_0)$ and $A$ is determined by the values $T(v_i)$ for $i \in \{1,\ldots,m\}$.
\end{proof}

\begin{proposition}\label{26.05.08}
If $B = \{b_0, b_1, \dots, b_{\ell-1}\} \subseteq \F_q$ forms an arithmetic progression with $b_0 = 0$, then the group of affine linear transformations $\operatorname{GA}(B, m)$ is isomorphic to $S_{m+1}$. 
\end{proposition}
\begin{proof}
Let $T({\bm x})$ be an element in $\operatorname{GA}(B, m)$ and consider the hyperplanes $H_i : x_i = 0$ for $1 \leq i \leq m$ and $H_0 : x_1 + \dots + x_m = b_{\ell - 1}$. Since the elements of $B$ form an arithmetic progression, it must be that the diagonal points of $B(m)$ lie on the hyperplane $H_0$. Thus, by Proposition~\ref{hyperlanes lemma}, the maximal hyperplanes over $B(m)$ are $H_0, \ldots, H_m$.

From Lemmas~\ref{26.05.15} and \ref{26.05.16},
we obtain 
\[|\operatorname{GA}(B,m)| \leq (m+1)!.\]
We now show that there are $(m+1)!$ distinct elements in $\operatorname{GA}(B, m)$ by describing how these transformations permute the vertices $v_j$'s.

\begin{itemize}
\item Type I (coordinate permutations). For each element $\sigma \in S_m$, the map
\[(x_1, \dots, x_m) \mapsto (x_{\sigma(1)}, \dots, x_{\sigma(m)})\]
fixes the origin $v_0 = \mathbf{0}$, permutes $\{v_1, \dots, v_m\}$, and fixes $B(m)$ because if $(b_{i_1}, b_{i_2}, \dots, b_{i_m}) \in B(m)$, then $(b_{i_{\sigma(1)}}, b_{i_{\sigma(2)}}, \dots, b_{i_{\sigma(m)}})$ is also in $B(m)$ as $i_{\sigma(1)} + \cdots + i_{\sigma(m)} = i_1+\cdots+i_m < \ell$. There are $m!$ Type I transformations.

\item Type II (corner flips). For each element $j \in \{1, \dots, m\}$, define $T_{j}({\bm x})$ by 
\begin{equation*}
x_j \mapsto b_{\ell - 1} - \sum_{i = 1}^m x_i, \qquad x_k \mapsto x_k \text{ for } k \neq j.
\end{equation*}
We now check that $T_j(B(m)) = B(m)$. Assume, without loss of generality, $j = 1$. Consider $\mathbf{b}=(b_{i_1},\dots,b_{i_m})\in B(m)$, i.e., $i_1 + \dots + i_m < \ell$. Therefore, 
\begin{align*}
T_1(\mathbf{b}) &= \left(b_{\ell -1} - (b_{i_1} + \cdots+b_{i_m}), b_{i_2}, \ldots, b_{i_m})\right)\\
&= \left(b_{\ell -1 - \left(i_1 + \cdots + i_m\right) }, b_{i_2}, \ldots, b_{i_m})\right) \in B(m).
\end{align*}
The last equality follows from the fact that as $B = \{b_0, b_1, \dots, b_{\ell-1}\}$ forms an arithmetic progression with $b_0 = 0$, then $b_i = ib_1$. Thus, $b_{i}-b_{j} = ib_1 - i^\prime b_1 = (i-i^\prime)b_1 = b_{i-i^\prime}$. With respect to the vertices, we can see that $T({\bm x})$ permutes $v_0$ and $v_j$ while fixing the rest. Together with the identity, we obtain $m+1$ transformations.
\end{itemize}

We can see that the Type I transformations form the stabilizer of the origin $v_0 = \mathbf{0}$. While a Type II transformation is a transposition $(v_0\; v_j)$. As any permutation $\tau$ of $\{v_0, v_1, \dots, v_m\}$ can be written as $\tau = \sigma \circ T_j({\bm x})$, where $T_j({\bm x})$ sends $v_0$ to $\tau(v_0) = v_j$, or the identity if $\tau$ fixes $v_0$, and $\sigma \in S_m$ handles the remaining permutation of $\{v_1, \dots, v_m\}$, we see that every permutation can be written in terms of Type I and Type II transformations. Combining the two bounds, we obtain $|\operatorname{GA}(B, m)| = (m+1)!$, and the action on vertices gives the isomorphism $\operatorname{GA}(B, m) \cong S_{m+1}$.
\end{proof}

\begin{theorem}\label{arth}
Assume that $b_0 = 0$ in $B = \{b_0, \dots, b_{\ell-1}\}$. The group $\operatorname{GA}(B,m)$ is isomorphic to $S_{m+1}$ if and only if $B$ is an arithmetic progression.
\end{theorem}
\begin{proof}
By Proposition~\ref{26.05.08}, we just need to show that if $\mathrm{GA}(B,m)$ is isomorphic to $S_{m+1}$, then $B$ is an arithmetic progression.

Assume that $\mathrm{GA}(B,m)$ is isomorphic to $S_{m+1}$. Define the vertices $v_0 := \mathbf{0}$ and $v_i := b_{\ell-1}e_i$ for $i=1,\ldots, m$. As there are $(m+1)!$ transformations, by Lemmas~\ref{26.05.15} and~\ref{26.05.16}, there is a transformation $T({\bm x}) = A{\bm x} + b$ that permutes the origin $v_0$ with $b_{\ell-1}e_1$ and fixes the rest of the vertices. The condition $T({\bm 0}) = b_{\ell-1}e_1$ implies that $b = b_{\ell-1}e_1$, and the condition $\phi(b_{\ell-1}e_1) = {\bm 0}$ implies that the first column $A^1$ of $A$ is $-e_1$. From $T(b_{\ell-1}e_i) = b_{\ell-1}e_i$ for $i=2,\ldots,m$, we get that the first row of $A$ has all $-1$ and $A$ has the identity matrix
$I_{m-1}$ in its lower-right submatrix. This is precisely a Type II
transformation, and explicitly
\[
  T(x_1,\ldots,x_m) = \Bigl(b_{\ell-1} - \textstyle\sum_{k=1}^m x_k,\;
  x_2,\;\ldots,\; x_m\Bigr).
\]

We now prove that $b_{\ell-1} - b_i = b_{\ell-1-i}$ for all
$0 \le i \le \ell-1$. As $(0, b_i, 0, \ldots, 0) \in B(m)$, then
\[T(0, b_i, 0, \ldots, 0) =  (b_{\ell-1} - b_i,\, b_i,\, 0, \ldots, 0) \in B(m).\]
Hence, $b_{\ell-1} - b_i = b_{\tau(i)}$ for some index $\tau(i)$ with $\tau(i) + i \le \ell-1$, i.e., $\tau(i) \le \ell-1-i$. Since $T(\bm x)$ is a bijection on $B(m)$, the map $\tau$ is a bijection on $\{0, \ldots, \ell-1\}$. Summing the
inequalities, we obtain
\[
  \sum_{i=0}^{\ell-1} \tau(i) \le \sum_{i=0}^{\ell-1}(\ell-1-i) =
  \frac{(\ell-1)\ell}{2},
\]
and since $\tau$ is a bijection on $\{0, \ldots, \ell-1\}$ both sides equal
$\frac{(\ell-1)\ell}{2}$. Hence equality holds at every $i$, forcing $\tau(i)
= \ell-1-i$, i.e., $b_{\ell-1} - b_i = b_{\ell-1-i}$ for all $0 \le i \le
\ell-1$.

Next, we must bound the values of $b_k - b_1$. For any $0 \leq i \leq \ell-2$, the point $(b_1, b_i, 0, \ldots, 0)$ is in $B(m)$. Applying $T$, we get:
\[
T(b_1, b_i, 0, \ldots, 0) = (b_{\ell-1} - b_i - b_1, b_i, 0, \ldots, 0) \in B(m).
\]
Substituting $b_{\ell-1} - b_i = b_{\ell-1-i}$, the first coordinate becomes $b_{\ell-1-i} - b_1$. Since this point is in $B(m)$, its first index must be at most $\ell - 1 -i $. Setting $k = \ell - 1 - i$, this implies that $b_k - b_1 = b_j$ for some $j \le k$. In other words, $b_k - b_1 \in \{b_0, b_1, \ldots, b_k\}$.

Now, we show $b_k - b_1 = b_{k-1}$ by induction. 
\begin{itemize}
    \item For $k=1$: $b_1 - b_1 = 0 = b_0$.
    \item For $k=2$: $b_2 - b_1 \in \{b_0, b_1, b_2\}$. It cannot be $b_0$ (since $b_2 \neq b_1$) and it cannot be $b_2$ (since $b_1 \neq 0$). Thus, $b_2 - b_1 = b_1 \implies b_2 = 2b_1$.
    \item For general $k$: The value $b_k - b_1$ must belong to $\{b_0, \ldots, b_k\}$. It cannot be $b_k$ (since $b_1 \neq 0$). It cannot equal $b_i$ for $i < k-1$ because that would mean $b_k = b_i + b_1 = b_{i+1}$ (by the induction hypothesis), which contradicts that the elements of $B$ are strictly distinct. Thus, by process of elimination, it must be that $b_k - b_1 = b_{k-1}$.
\end{itemize}

Therefore, $b_k = b_{k-1} + b_1$ for all $k$, meaning $b_k = k \cdot b_1$. This concludes the proof that $B$ is an arithmetic progression with common difference $b_1$.
\end{proof}

The following result characterizes the group $\operatorname{GA}(B,m)$ when $b_0=0$ in $B$.
\begin{theorem}\label{26.05.18}
Assume that $b_0 = 0$ in $B = \{b_0, \dots, b_{\ell-1}\}$. We have
\[
\mathrm{GA}(\mathcal{B},m) \cong
\begin{cases}
S_{m+1} & \text{ if } B \text{ is an arithmetic progression, and}\\
S_m & \text{ otherwise.}
\end{cases}
\]
\end{theorem}
\begin{proof}
The coordinate permutations $(x_1,\dots,x_m) \mapsto (x_{\sigma(1)}, \dots, x_{\sigma(m)})$ preserve the index-sum condition defining $\mathcal{B}(m)$, so they lie in $\mathrm{GA}(\mathcal{B},m)$ and give an embedded copy of $S_m$ stabilizing $v_0 = 0$. When the plane $H_0$ defined in Proposition~\ref{hyperlanes lemma} is maximal, any $T({\bm x}) \in \mathrm{GA}(\mathcal{B},m)$ permutes the $m+1$ vertices by Lemma~\ref{26.05.15}. Therefore, by Lemma~\ref{26.05.16}, we obtain an injection $\mathrm{GA}(\mathcal{B},m) \hookrightarrow S_{m+1}$ whose image contains $S_m$, forcing the image to be $S_m$ or $S_{m+1}$ by the maximality of $S_m$ into $S_{m+1}$.
\end{proof}
As a consequence of the previous results, we obtain the permutation group of CAP codes for certain cases.
\begin{corollary}
Assume that $b_0 = 0$ in $B = \{b_0, \dots, b_{\ell-1}\}$ and $0 < \nu < \ell/2$. Then,
\[
\Per(\CAP(B,m,\nu)) \cong
\begin{cases}
S_{m+1} & \text{ if } B \text{ is an arithmetic progression, and}\\
S_m & \text{ otherwise.}
\end{cases}
\]
\end{corollary}
\begin{proof}
This is a consequence of Theorems~\ref{26.05.17} and~\ref{26.05.18}.
\end{proof}

\section*{Declarations}
\subsection*{Conflict of interest} The authors declare no conflict of interest.


\begin{thebibliography}{10}

\bibitem{Ball2020CodingTP}
T.~Ball, E.~Camps, H.~Chimal-Dzul, D.~Jaramillo-Velez, H.~H. L'opez, N.~S. Nichols, M.~Perkins, I.~Soprunov, G.~Vera-Mart'inez, and G.~R. Whieldon.
\newblock Coding theory package for macaulay2.
\newblock {\em ArXiv}, abs/2007.06795, 2020.

\bibitem{LESS_is_even_more}
L.~Beckwith, A.~Esser, E.~Persichetti, P.~Santini, and F.~Zweydinger.
\newblock {LESS} is even more: Optimizing digital signatures from code equivalence.
\newblock Cryptology {ePrint} Archive, Paper 2025/1424, 2025.

\bibitem{beelenGHWcartesian}
P.~Beelen and M.~Datta.
\newblock Generalized {H}amming weights of affine {C}artesian codes.
\newblock {\em Finite Fields Appl.}, 51:130--145, 2018.

\bibitem{LESS}
J.-F. Biasse, G.~Micheli, E.~Persichetti, and P.~Santini.
\newblock Less is more: Code-based signatures without syndromes.
\newblock In {\em Progress in Cryptology - AFRICACRYPT 2020: 12th International Conference on Cryptology in Africa, Cairo, Egypt, July 20 – 22, 2020, Proceedings}, page 45–65, Berlin, Heidelberg, 2020. Springer-Verlag.

\bibitem{bioglio}
V.~Bioglio, I.~Land, and C.~Pillet.
\newblock Group properties of polar codes for automorphism ensemble decoding.
\newblock {\em IEEE Transactions on Information Theory}, 69(6):3731--3747, 2023.

\bibitem{buchberger}
B.~Buchberger.
\newblock Ein {A}lgorithmus zum {A}uffinden der {B}asiselemente des {R}estklassenringes nach einem nulldimensionalen {P}olynomideal.
\newblock Dissertation an dem Math. Inst. der Universit\"at von Innsbruck, 1965.

\bibitem{eduardoGHWHyperbolic}
E.~Camps-Moreno, I.~Garc\'{\i}a-Marco, H.~H. L\'{o}pez, I.~M\'{a}rquez-Corbella, E.~Mart\'{\i}nez-Moro, and E.~Sarmiento.
\newblock On the generalized {H}amming weights of hyperbolic codes.
\newblock {\em Journal of Algebra and Its Applications}, 23(07):2550062, 2024.

\bibitem{sanjoseGHWNT}
E.~Camps-Moreno, H.~H. L\'opez, G.~L. Matthews, and R.~San-Jos\'e.
\newblock The weight hierarchy of decreasing norm-trace codes.
\newblock {\em Des. Codes Cryptogr.}, 93(7):2873--2894, 2025.

\bibitem{sanjoseGHWsquarefree}
C.~Carvalho, H.~López, and R.~San-José.
\newblock Cartesian square-free codes.
\newblock {\em ArXiv 2511.08304}, 2025.

\bibitem{nestedcartesian}
C.~Carvalho, V.~G.~L. Neumann, and H.~H. L\'{o}pez.
\newblock Projective nested cartesian codes.
\newblock {\em Bull. Braz. Math. Soc. (N.S.)}, 48(2):283--302, 2017.

\bibitem{chen2007secure}
H.~Chen, R.~Cramer, S.~Goldwasser, R.~De~Haan, and V.~Vaikuntanathan.
\newblock Secure computation from random error correcting codes.
\newblock In {\em Annual International Conference on the Theory and Applications of Cryptographic Techniques}, pages 291--310. Springer, 2007.

\bibitem{geramita}
S.~M. Cooper, A.~Seceleanu, c.~O. Toh\u{a}neanu, M.~V. Pinto, and R.~H. Villarreal.
\newblock Generalized minimum distance functions and algebraic invariants of {G}eramita ideals.
\newblock {\em Adv. in Appl. Math.}, 112:101940, 34, 2020.

\bibitem{cox}
D.~A. Cox, J.~Little, and D.~O'Shea.
\newblock {\em Ideals, varieties, and algorithms}.
\newblock Undergraduate Texts in Mathematics. Springer, Cham, fourth edition, 2015.
\newblock An introduction to computational algebraic geometry and commutative algebra.

\bibitem{fengraoMajority}
G.~L. Feng and T.~R.~N. Rao.
\newblock Decoding algebraic-geometric codes up to the designed minimum distance.
\newblock 39(1):37--45, 1993.

\bibitem{ghorpade}
S.~R. Ghorpade.
\newblock A note on {N}ullstellensatz over finite fields.
\newblock {\em Contemp. Math.}, 738:23--32, 2019.

\bibitem{villarreal_gorenstein_prm_type}
M.~Gonz\'alez-Sarabia, H.~Mu\~noz George, J.~A. Ordaz, E.~S\'aenz-de Cabez\'on, and R.~H. Villarreal.
\newblock Indicator functions, v-numbers and {G}orenstein rings in the theory of projective {R}eed-{M}uller-type codes.
\newblock {\em Des. Codes Cryptogr.}, 92(11):3317--3353, 2024.

\bibitem{guruswammiGHWlistdecodingTensorInterleaved}
P.~Gopalan, V.~Guruswami, and P.~Raghavendra.
\newblock List decoding tensor products and interleaved codes.
\newblock {\em SIAM J. Comput.}, 40(5):1432--1462, 2011.

\bibitem{Grassl_Roetteler}
M.~Grassl and M.~Roetteler.
\newblock Leveraging automorphisms of quantum codes for fault-tolerant quantum computation.
\newblock In {\em 2013 IEEE International Symposium on Information Theory}, pages 534--538, 2013.

\bibitem{M2}
D.~R. Grayson and M.~E. Stillman.
\newblock Macaulay2, a software system for research in algebraic geometry.
\newblock Available at \url{http://www.math.uiuc.edu/Macaulay2/}.

\bibitem{guruswammiGHWlistdecoding}
V.~Guruswami.
\newblock List decoding from erasures: bounds and code constructions.
\newblock {\em IEEE Trans. Inform. Theory}, 49(11):2826--2833, 2003.

\bibitem{hamadaSteaneEnlargement}
M.~Hamada.
\newblock Concatenated quantum codes constructible in polynomial time: efficient decoding and error correction.
\newblock {\em IEEE Trans. Inform. Theory}, 54(12):5689--5704, 2008.

\bibitem{pellikaanGHWRM}
P.~Heijnen and R.~Pellikaan.
\newblock Generalized {H}amming weights of {$q$}-ary {R}eed-{M}uller codes.
\newblock {\em IEEE Trans. Inform. Theory}, 44(1):181--196, 1998.

\bibitem{Huffman_Pless_2003}
W.~C. Huffman and V.~Pless.
\newblock {\em Fundamentals of Error-Correcting Codes}.
\newblock Cambridge University Press, 2003.

\bibitem{jaramillo}
D.~Jaramillo, M.~Vaz~Pinto, and R.~H. Villarreal.
\newblock Evaluation codes and their basic parameters.
\newblock {\em Des. Codes Cryptogr.}, 89(2):269--300, 2021.

\bibitem{1Johannsen_polar}
L.~Johannsen, C.~Kestel, M.~Geiselhart, T.~Vogt, S.~T. Brink, and N.~Wehn.
\newblock Successive cancellation automorphism list decoding of polar codes.
\newblock In {\em 2023 12th International Symposium on Topics in Coding (ISTC)}, pages 1--5, 2023.

\bibitem{polar_19}
M.~Kamenev, Y.~Kameneva, O.~Kurmaev, and A.~Maevskiy.
\newblock Permutation decoding of polar codes.
\newblock In {\em 2019 XVI International Symposium "Problems of Redundancy in Information and Control Systems" (REDUNDANCY)}, pages 1--6, 2019.

\bibitem{kasamiRM}
T.~Kasami, S.~Lin, and W.~W. Peterson.
\newblock New generalizations of the {R}eed-{M}uller codes. {I}. {P}rimitive codes.
\newblock {\em IEEE Trans. Inform. Theory}, IT-14:189--199, 1968.

\bibitem{kkks}
A.~Ketkar, A.~Klappenecker, S.~Kumar, and P.~K. Sarvepalli.
\newblock Nonbinary stabilizer codes over finite fields.
\newblock {\em IEEE Trans. Inform. Theory}, 52(11):4892--4914, 2006.

\bibitem{kopparty_high_rate}
S.~Kopparty, M.~Kumar, and H.~Sha.
\newblock High rate multivariate polynomial evaluation codes.
\newblock In {\em S{TOC}'25---{P}roceedings of the 57th {A}nnual {ACM} {S}ymposium on {T}heory of {C}omputing}, pages 810--821. ACM, New York, [2025] \copyright 2025.

\bibitem{RM_CA1}
S.~Kudekar, S.~Kumar, M.~Mondelli, H.~D. Pfister, E.~\c{S}a\c{s}o\u{g}lu, and R.~Urbanke.
\newblock {R}eed-{M}uller codes achieve capacity on erasure channels.
\newblock In {\em Proceedings of the Forty-Eighth Annual ACM Symposium on Theory of Computing}, STOC '16, page 658–669, New York, NY, USA, 2016. Association for Computing Machinery.

\bibitem{RM_CA2}
S.~Kudekar, S.~Kumar, M.~Mondelli, H.~D. Pfister, E.~Şaşoǧlu, and R.~L. Urbanke.
\newblock {R}eed–{M}uller codes achieve capacity on erasure channels.
\newblock {\em IEEE Transactions on Information Theory}, 63(7):4298--4316, 2017.

\bibitem{matsumotoRGHW}
J.~Kurihara, T.~Uyematsu, and R.~Matsumoto.
\newblock Secret sharing schemes based on linear codes can be precisely characterized by the relative generalized hamming weight.
\newblock {\em IEICE Trans. Fundam. Electron. Commun. Comput. Sci.}, E95.A(11):2067--2075, 2012.

\bibitem{Lopez2014-bq}
H.~H. L{\'o}pez, C.~Renter{\'\i}a-M{\'a}rquez, and R.~H. Villarreal.
\newblock Affine cartesian codes.
\newblock {\em Des. Codes Cryptogr.}, 71(1):5--19, Apr. 2014.

\bibitem{Lopez2021-tu}
H.~H. L{\'o}pez, I.~Soprunov, and R.~H. Villarreal.
\newblock The dual of an evaluation code.
\newblock {\em Des. Codes Cryptogr.}, 89(7):1367--1403, July 2021.

\bibitem{luoPropertiesRGHWs}
Y.~Luo, C.~Mitrpant, A.~J.~H. Vinck, and K.~Chen.
\newblock Some new characters on the wire-tap channel of type {II}.
\newblock {\em IEEE Trans. Inform. Theory}, 51(3):1222--1229, 2005.

\bibitem{LOPEZ201913}
H.~H. López, F.~Manganiello, and G.~L. Matthews.
\newblock Affine cartesian codes with complementary duals.
\newblock {\em Finite Fields and Their Applications}, 57:13--28, 2019.

\bibitem{munuera1999hermitian}
C.~Munuera and D.~Ramirez.
\newblock The second and third generalized hamming weights of hermitian codes.
\newblock {\em IEEE Transactions on Information Theory}, 45(2):709--712, 1999.

\bibitem{GHWsToricSquarefree}
N.~Patanker and S.~K. Singh.
\newblock Generalized {H}amming weights of toric codes over hypersimplices and squarefree affine evaluation codes.
\newblock {\em Adv. Math. Commun.}, 17(3):626--643, 2023.

\bibitem{sanjoseGHWMPC}
R.~San-Jos\'e.
\newblock About the generalized {H}amming weights of matrix-product codes.
\newblock {\em Comput. Appl. Math.}, 44(4):Paper No. 186, 2025.

\bibitem{sanjoseGHWsPackage}
R.~San-Jos\'e.
\newblock An algorithm for computing generalized {H}amming weights and the {S}age package {$\tt{GHWs}$}.
\newblock {\em ACM Trans. Math. Software}, 51(4):Art. 26, 20, 2025.

\bibitem{githubGHWs}
R.~San-Jos\'{e}.
\newblock {GHWs}: A {S}age package for computing the generalized {H}amming weights of a linear code. {G}it{H}ub repository.
\newblock Available online: \url{https://github.com/RodrigoSanJose/GHWs}, 2025.

\bibitem{sorensen}
A.~B. S{\o}rensen.
\newblock Projective {R}eed-{M}uller codes.
\newblock {\em IEEE Trans. Inform. Theory}, 37(6):1567--1576, 1991.

\bibitem{sagemath}
{The Sage Developers}.
\newblock {\em {S}ageMath, the {S}age {M}athematics {S}oftware {S}ystem ({V}ersion 10.9)}, 2026.
\newblock {\tt https://www.sagemath.org}.

\bibitem{vardyIntractability}
A.~Vardy.
\newblock The intractability of computing the minimum distance of a code.
\newblock {\em IEEE Trans. Inform. Theory}, 43(6):1757--1766, 1997.

\bibitem{villarreal_book_monomial_algebras}
R.~H. Villarreal.
\newblock {\em Monomial algebras}.
\newblock Monographs and Research Notes in Mathematics. CRC Press, Boca Raton, FL, third edition, 2026.

\bibitem{wei1991ghw}
V.~Wei.
\newblock Generalized hamming weights for linear codes.
\newblock {\em IEEE Transactions on Information Theory}, 37(5):1412--1418, 1991.

\end{thebibliography}

\end{document}